\documentclass[11pt,letterpaper]{article}
\usepackage{amsmath}
\usepackage{amssymb}
\usepackage{amsthm}
\usepackage{mathrsfs}
\usepackage{xcolor}
\usepackage{enumitem}
\usepackage{lscape}
\usepackage{mathtools}
\usepackage{multirow}
\usepackage{tikz-qtree}
\usepackage{pgf,tikz}
\usepackage{hyperref}
\usepackage{multicol}

\newcommand\commentout[1]{}

\newtheorem{procedure}{Procedure}{\bfseries}{\itshape}
\newtheorem{definition}{Definition}{\bfseries}{\itshape}
\newtheorem{example}{Example}{\bfseries}{\itshape}
\newtheorem{theorem}{Theorem}{\bfseries}{\itshape}
\newtheorem{lemma}{Lemma}{\bfseries}{\itshape}

\newenvironment{varlem}[1][Lemma]{\begin{trivlist}
\item[\hskip \labelsep {\bfseries #1}]}{\end{trivlist}}

\newcommand{\bbN}{\mathbb{N}}

\newcommand{\calA}{\mathcal{A}}
\newcommand{\calB}{\mathcal{B}}
\newcommand{\calC}{\mathcal{C}}
\newcommand{\calD}{\mathcal{D}}

\newcommand{\calF}{\mathcal{F}}

\newcommand{\calL}{\mathcal{L}}
\newcommand{\calM}{\mathcal{M}}

\newcommand{\calR}{\mathcal{R}}

\newcommand{\calT}{\mathcal{T}}

\newcommand{\scrG}{\mathscr{G}}

\newcommand{\scrL}{\mathscr{L}}

\DeclareMathOperator{\Projh}{\mathrm{Proj_{h}}}
\DeclareMathOperator{\Projr}{\mathrm{Proj_{r}}}
\DeclareMathOperator{\Succ}{\mathit{Succ}}
\DeclareMathOperator{\range}{\mathrm{range}}
\DeclareMathOperator{\GC}{\mathrm{GC}}
\DeclareMathOperator{\CB}{\mathrm{CB}}
\DeclareMathOperator{\GB}{\mathrm{GB}}
\DeclareMathOperator{\B}{\mathrm{B}}
\DeclareMathOperator{\Parity}{\mathrm{P}}
\DeclareMathOperator{\Rabin}{\mathrm{R}}
\DeclareMathOperator{\muR}{\mathrm{\mu R}}
\DeclareMathOperator{\muGC}{\mathrm{\mu GC}}
\DeclareMathOperator{\Streett}{\mathrm{S}}

\DeclareMathOperator{\CA}{\mathcal{CA}}
\DeclareMathOperator{\gc}{\mathrm{gc}}
\DeclareMathOperator{\cb}{\mathrm{cb}}
\DeclareMathOperator{\even}{\mathrm{even}}
\DeclareMathOperator{\odd}{\mathrm{odd}}
\DeclareMathOperator{\Inf}{\mathrm{Inf}}
\DeclareMathOperator{\Cover}{\mathrm{Cover}}
\DeclareMathOperator{\Mini}{\mathrm{Mini}}
\DeclareMathOperator{\width}{\mathit{width}}

\DeclareMathOperator{\TOP}{\mathrm{TOP}}
\DeclareMathOperator{\TOPs}{\mathrm{TOPs}}
\DeclareMathOperator{\ITS}{\mathrm{ITS}}
\DeclareMathOperator{\MWP}{\mathrm{MWP}} 

\begin{document}

\title{Tight Upper Bounds for Streett and Parity Complementation}

\author{Yang Cai\\
MIT CSAIL \\
The Stata Center, 32-G696 \\
Cambridge, MA 02139 USA \\
ycai@csail.mit.edu \\
\and
Ting Zhang \\
Iowa State University \\
226 Atanasoff Hall \\
Ames, IA 50011 USA\\
tingz@iastate.edu
}

\date{}
\maketitle
\thispagestyle{empty}

\begin{abstract}
Complementation of finite automata on infinite words is not only a fundamental problem in automata theory, but also serves as a cornerstone for solving numerous decision problems in mathematical logic, model-checking, program analysis and verification. For Streett complementation, a significant gap exists between the current lower bound $2^{\Omega(n\lg nk)}$ and upper bound $2^{O(nk\lg nk)}$, where $n$ is the state size, $k$ is the number of Streett pairs, and $k$ can be as large as $2^{n}$. Determining the complexity of Streett complementation has been an open question since the late '80s. In this paper show a complementation construction with upper bound $2^{O(n \lg n+nk \lg k)}$ for $k = O(n)$ and $2^{O(n^{2} \lg n)}$ for $k = \omega(n)$, which matches well the lower bound obtained in~\cite{CZ11a}. We also obtain a tight upper bound $2^{O(n \lg n)}$ for parity complementation.
\end{abstract}

\setlength{\columnsep}{-10pt}

\section{Introduction}
\label{sec:introduction}

Automata on infinite words ($\omega$-automata) have wide applications in synthesis and verification of reactive concurrent systems. Complementation plays a fundamental role in many of these applications, especially in solving the language containment problem: whether a language recognized by automaton $\calA$ is contained by another language represented by automaton $\calB$, which is equivalent to whether the language of $\calA$ and the complementary language of $\calB$ intersect. In automata-theoretic model checking~\cite{Kur94,VW86}, both system behaviors and logical specifications are represented as formal languages, and model checking by and large amounts to solving the corresponding language containment problem. As both language intersection and emptiness test are rather easy, the efficiency of complementation becomes crucial to practical deployment of model-checking tools. For this reason and many others, determining the state complexity of the complementation problem has been extensively studied in the last four decades~\cite{Var07}.

\paragraph{Related Work.} $\omega$-automata were invented by B\"{u}chi in 1962 as a method of attack on definability and decision problems for monadic second order logic on arithmetics (S1S)~\cite{Buc62}. That type of $\omega$-automata are nowadays called B\"{u}chi automata. The initial B\"{u}chi complementation was not explicitly constructive and required double exponential blow-up~\cite{Buc62}. But since then B\"{u}chi complementation has been extensively studied. The upper bound was continuously improved to $2^{O(n^{2})}$~\cite{SVW87}, $2^{O(n\lg n)}$~\cite{Saf88}, $O((6n)^{n})$~\cite{KV01}, $O((0.97n)^{n})$~\cite{FKV06} and finally to $O(n^{2}L(n))$ where $L(n)\approx(0.76n)^{n}$~\cite{Sch09}, which matched well the lower bound $\Omega(L(n))$~\cite{Yan06}.

Complementation for automata with rich acceptance conditions, such as Rabin automata and Streett automata, is much more sophisticated. Kupferman and Vardi showed a $2^{O(nk \lg n)}$ complementation construction for Rabin automata~\cite{KV05a}, and we showed this construction is essentially optimal~\cite{CZL09}. This leaves Streett complementation the last classical problem where the gap between the lower and upper bounds is substantial. Besides that, Streett complementation has an importance of its own. Streett automata share identical algebraic structures with B\"{u}chi automata, except being equipped with richer acceptance conditions. A Streett acceptance condition comprises a finite list of indexed pairs of sets of states. Each pair consists of an enabling set and a fulfilling set. A run is accepting if for each pair, if the run visits states in the enabling set infinitely often, then it also visits states in the fulfilling set infinitely often. This naturally corresponds to the \emph{strong fairness} condition that infinitely many requests are responded infinitely often, a necessary requirement for meaningful computations~\cite{FK84,Fra86}. Another advantage of Streett automata is that they are much more succinct than B\"{u}chi automata; it is unavoidable in the worst case to have $2^{n}$ state blow-up to translate a Streett automaton with $O(n)$ states and $O(n)$ index pairs to an equivalent B\"{u}chi automaton~\cite{SV89}. An interesting question is: to what extent does the gain from the succinctness have to be paid back at the time of complementation?

The first construction for Streett complementation was given by Safra and Vardi, and that construction required $2^{O((nk)^{5})}$ state blow-up~\cite{SV89}. Klarlund improved this bound to $2^{O(nk \lg nk)}$~\cite{Kla91}. The same bound was achieved by Safra via determinization~\cite{Saf92}, by Piterman with an improved determinization construction~\cite{Pit06}, and by Kupferman and Vardi~\cite{KV05a}. However, so far no construction has been proved to cost less than $2^{O(nk \lg nk)}$ states. The question of whether Streett complementation can be further improved from $2^{O(nk \lg nk)}$ has been constantly raised in the recent literature~\cite{KV05a,Yan06,Var07}. In this paper we answer this question affirmatively.

\paragraph{Ranking-based Complementation.}
A Ranking-based complementation was first proposed by Klarlund~\cite{Kla91}. Klarlund's B\"{u}chi complementation (resp. Streett complementation) relies on \emph{quasi co-B\"{u}chi measure} (resp. \emph{quasi Rabin measure}), which is a ranking function on states in a run graph, measuring the progress of a run toward being accepted. By this complementation scheme, Klarlund gave a $2^{O(n\lg n)}$ B\"{u}chi complementation and a $2^{O(nk\lg nk)}$ Streett complementation~\cite{Kla91}. Kupferman and Vardi developed a similar idea into an elegant and comprehensive framework~\cite{KV05b,Kup06}, obtaining complementation constructions for B\"{u}chi~\cite{KV01}, generalized B\"{u}chi~\cite{KV04}, Rabin and Streett~\cite{KV05a}.

\paragraph{Our Results.}
Our Streett complementation is obtained by improving Kupferman and Vardi's construction in~\cite{KV05a}. We show that the larger the Rabin index size $k$, the higher the correlation between infinite paths in a run graph satisfying a universal Rabin condition (the dual of an existential Streett condition), and characterize the correlation using two tree structures: $\ITS$ (\emph{Increasing Tree of Sets}) and $\TOP$ (\emph{Tree of Ordered Partitions}), both with elegant combinatorial properties. We show that our construction renders a upper bound $U(n,k)$, which is $2^{O(n \lg n+nk \lg k)}$ for $k = O(n)$ and $2^{O(n^{2} \lg n)}$ for $k = \omega(n)$. $U(n,k)$ is a significant improvement from the previous best bound when $k=\omega(n)$. Speaking loosely, we gain succinctness without paying a dramatically higher price for complementation. $U(n,k)$ also matches the lower bound $L(n,k)$, which is $2^{\Omega(n \lg n+nk \lg k)}$ for $k = O(n)$ and $2^{\Omega(n^{2} \lg n)}$ for $k = \omega(n)$~\cite{CZ11a}. By a similar technique, we also obtain a $2^{O(n \lg n)}$ upper bound for parity complementation, which is essentially optimal, as parity automata generalizes B\"{u}chi automata, whose complementation lower bound is $2^{\Omega(n \lg n)}$~\cite{Mic88,Lod99}. This is surprising as the index size $k$ (though small as $k \le \lfloor (n+1)/2\rfloor$) has no appearance in the asymptotical bound. We believe this is of practical interest as well, because it tells us that parity automata provide a richer acceptance condition without incurring an asymptotically higher cost on complementation. Combining the result with the one in~\cite{CZ11a} and previous findings in the literature, we now have a complete characterization of complementation complexity for $\omega$-automata of common types. Figure~\ref{fig:bound-summary} summarizes these results.

\begin{figure}[th!]
\begin{center}
\begin{tabular}{|l||ll|l|l|}
\hline
Type        & Bound                  &                       & Lower            & Upper        \\ \hline
B\"{u}chi        & $2^{\Theta(n \lg n)}$  &                  & \cite{Mic88}     & \cite{Saf88} \\ \hline
Generalized B\"{u}chi      & $2^{\Theta(n \lg nk)}$ & $k = O(2^{n})$   & \cite{Yan06}     & \cite{KV05a} \\ \hline
\multirow{2}{*}{Streett}  & $2^{\Theta(n \lg n + nk \lg k)}$ & $k = O(n)$      & \multirow{2}{*}{\cite{CZ11a}}  & \multirow{2}{*}{this} \\
                             & $2^{\Theta(n^{2} \lg n)}$ & $k=\omega(n)$   &      &       \\ \hline
Rabin    & $2^{\Theta(nk \lg n)}$ & $k = O(2^{n})$   & \cite{CZL09}     & \cite{KV05a} \\ \hline
Parity   & $2^{\Theta(n \lg n)}$  & $k = O(n)$       & \cite{Mic88}     & this         \\ \hline
\end{tabular}
\end{center}
\caption{\textrm{The complementation complexities for $\omega$-automata of common types. Note that co-B\"{u}chi and generalized co-B\"{u}chi automata are not listed here because they are not complete for $\omega$-regular languages, i.e., the class of $\omega$-languages that can be recognized by co-B\"{u}chi automata (or generalized co-B\"{u}chi automata) is a proper subclass of $\omega$-regular languages.}}
\label{fig:bound-summary}
\end{figure}

\paragraph{Paper Organization.}
Section~\ref{sec:preliminaries} introduces basic notations and terminology in automata theory. Section~\ref{sec:ranking-based-complementation} presents the framework of ranking based complementation; it introduces B\"{u}chi complementation~\cite{KV01}, generalized B\"{u}chi complementation~\cite{KV04} and Streett complementation~\cite{KV05a}. Section~\ref{sec:streett-complementation} presents our Streett complementation construction and Section~\ref{sec:complexity} proves its complexity. Section~\ref{sec:parity-complementation} establishes a tight upper bound for parity complementation. Section~\ref{sec:conclusion} concludes with a discussion of future work. All proofs are placed in the appendix. 
\section{Preliminaries}
\label{sec:preliminaries}

\paragraph{Basic Notations.}
Let $\bbN$ denote the set of natural numbers. We write $[i..j]$ for $\{ k \in \bbN \, \mid \, i \le k \le j\}$, $[i..j)$ for $[i..j-1]$, $[n]$ for $[0..n)$, and $[n]^{\even}$ and $[n]^{\odd}$ for even numbers and odd numbers in $[n]$, respectively. For an infinite sequence $\varrho$, we use $\varrho(i)$ to denote the $i$-th component for $i \in \bbN$. For a finite sequence $\alpha$, we use $|\alpha|$ to denote the length of $\alpha$, $\alpha[i]$ ($i \in [1..|\alpha|]$) to denote the object at the $i$-th position, and $\alpha[i..j]$ (resp. $\alpha[i..j)$) to denote the subsequence of $\alpha$ from position $i$ to position $j$ (resp. $j-1$). When we compare finite sequences of numbers, $>_{m}, \ge_{m}, =_{m}, <_{m}, \le_{m}$ mean the corresponding standard lexicographical orderings up to position $m$. We reserve $n$ and $k$ as parameters of complementation instances ($n$ for state size and $k$ for index size), and define $\mu=\min(n,k)$ and $I=[1..k]$.

%%%%%%%%%%%%%%%%%%%%%%%%%%%%%%%%%%%%%%%%%%%%%%%%%
% Definition of nondeterministic automata
%%%%%%%%%%%%%%%%%%%%%%%%%%%%%%%%%%%%%%%%%%%%%%%%%

\paragraph{Automata and Runs.}
A finite automaton on infinite words ($\omega$-automaton) is a tuple $\calA=(\Sigma, Q, Q_{0}, \Delta, \calF)$ where $\Sigma$ is an alphabet, $Q$ is a finite set of states, $Q_{0} \subseteq Q$ is a set of initial states, $\Delta \subseteq Q \times \Sigma \times Q$ is a set of transitions, and $\calF$ is an acceptance condition.

An infinite word ($\omega$-word) over $\Sigma$ is an infinite sequence of letters in $\Sigma$. A \emph{run} $\varrho$ of $\calA$ over an $\omega$-word $w$ is an infinite sequence of states in $Q$ such that $\varrho (0)\in Q_{0}$
and, $\langle\varrho(i),w(i),\varrho(i\!+\!1)\rangle \in \Delta$ for $i \in \bbN$. Let $\Inf(\varrho)$ be the set of states that occur infinitely many times in $\varrho$. An automaton accepts $w$ if a run $\varrho$ over $w$ exists that satisfies $\calF$, which is usually defined as a predicate on $\Inf(\varrho)$. The language of $\calA$, written $\scrL(\calA)$, is the set of $\omega$-words accepted by $\calA$.

%%%%%%%%%%%%%%%%%%%%%%%%%%%%%%%%%%%%%%%%%%%%%%%%
% Acceptance conditions
%%%%%%%%%%%%%%%%%%%%%%%%%%%%%%%%%%%%%%%%%%%%%%%%
\paragraph{Acceptance Conditions and Types.}
$\omega$-automata are classified according to their acceptance conditions. Below we list automata of common types. Let $G$ and $B$ be functions from $I$ to $2^{Q}$.
\begin{itemize}
\item \emph{Generalized B\"{u}chi}: $\langle B \rangle_{I}$: $\forall i \in I, \Inf(\varrho) \cap B(i) \neq \emptyset$.
\item \emph{B\"{u}chi}: $\langle B \rangle_{I}$ with $I=\{1\}$ (i.e., $k=1$).
\item \emph{Streett}: $\langle G,B\rangle_{I}$: $\forall i \in I$, $\Inf(\varrho) \cap G(i)\neq\emptyset \to \Inf(\varrho)\cap B(i)\neq\emptyset$.
\item \emph{Parity}: $\langle G,B\rangle_{I}$ with $B(1) \subset G(1) \subset \cdots \subset B(k) \subset G(k)$.
\item \emph{Generalized co-B\"{u}chi}: $[B]_{I}$: $\exists i \in I$, $\Inf(\varrho) \cap B(i) = \emptyset$.
\item \emph{Co-B\"{u}chi}: $[B]_{I}$ with $I=\{1\}$ (i.e., $k=1$).
\item \emph{Rabin}: $[G,B]_{I}$: $\exists i \in I$, $\Inf(\varrho)\cap G(i)\neq\emptyset \wedge \Inf(\varrho)\cap B(i)=\emptyset$.
\end{itemize}
We use $\GB$, $\B$, $\Streett$, $\Parity$, $\GC$, $\CB$, and $\Rabin$, respectively, to denote the above acceptance conditions. By $T$-automata we mean the $\omega$-automata with $T$-condition. Note that $\B$ and $\CB$, $\GB$ and $\GC$, and $\Streett$ and $\Rabin$ are dual to each other, respectively. Also note that generalized B\"{u}chi and parity automata are both subclasses of Streett automata, and so are generalized co-B\"{u}chi and parity automata to Rabin automata. Let $J \subseteq I$. We use $[G,B]_{J}$ to denote the Rabin condition with respect to only indices in $J$. When $J$ is a singleton, say $J=\{j\}$, we simply write $[G(j),B(j)]$ for $[G,B]_{J}$. The same convention is used for other conditions. For a Streett condition $\langle G,B\rangle_{I}$, we can assume that $B$ is injective, because if $B(i)=B(i')$ for two different $i,i' \in I$, then we can replace $\langle G,B\rangle_{\{i,i'\}}$ by $\langle G(i)\cup G(i'), B(i) \rangle$. The same assumption is made for any Rabin condition $[G, B]_{I}$.

\paragraph{$\Delta$-Graphs.}
A \emph{$\Delta$-graph} of an $\omega$-word $w$ under $\calA$ is a directed graph $\scrG_{w}=(V,E)$ where $V=Q \times \bbN$ and $E=\{\langle\langle q,l\rangle,\langle q',l+1\rangle\rangle \in V \times V \, \mid \, q, q' \in Q,\ i \in \bbN, \langle q,w(i),q'\rangle\in\Delta \, \}$. By the \emph{$i$-th level}, we mean the vertex set $Q \times \{i\}$. Let $S$ be a subset of $Q$. We call a vertex $v=\langle q, l\rangle$ \emph{$S$-vertex} if $q \in S$. When level index is of no importance in the context, we use $q$ and $v$ interchangeably. In particular, by an abuse of notation we write $v \in S$ to mean $v = \langle q, l\rangle$ for some $l \in \bbN$ and $q \in S$. $\Delta$-graphs for finite words are similarly defined. The length of a finite $\Delta$-graph is the number of levels minus $1$. By unit $\Delta$-graphs we mean $\Delta$-graphs of length $1$. A unit $\Delta$-graph encodes all possible transitions upon reading a letter. By \emph{width} of $\scrG_{w}$ (written $\width(\scrG_{w}$)) we mean the maximum number of pairwise non-intersecting infinite paths in $\scrG_{w}$. Clearly, for any $w$, $\width(\scrG_{w}) \le |Q|$. 
\section{Ranking-based Complementation}
\label{sec:ranking-based-complementation}

In this section we introduce ranking-based complementation constructions developed by Klarlund, Kupferman and Vardi~\cite{Kla91,KV01,KV05b,Kup06}. Note that all complexity related notions are parameterized with $n$ and $k$, but we do not list them explicitly unless required for clarity. We adopt the following naming convention: when we talk about behaviors of a source automaton, a $T$-condition means an \emph{existential} one (i.e., a path in a $\Delta$-graph that satisfies $T$), while in the context of complementation, a $T$-condition means a \emph{universal} one (i.e., every path in a $\Delta$-graph satisfies $T$).

\paragraph{Ranking-based Complementation Scheme.}
Let $\calA$ be a $T$-automaton and $\CA$ a purported B\"{u}chi automaton that complements $\calA$. An $\omega$-word $w$ is accepted by $\calA$ if and only if the $\Delta$-graph $\scrG_{w}$ contains an infinite path that satisfies the $T$-condition. Consequently, $w$ is accepted by $\CA$ if and only if \emph{all} paths in $\scrG_{w}$ satisfy the dual co-$T$ condition (for short, $\scrG_{w}$ is co-$T$ accepting). Complementation essentially amounts to transforming a \emph{universal} co-$T$ condition into an \emph{existential} B\"{u}chi condition. Rankings on $\Delta$-graphs provide a solution; $\scrG_{w}$ satisfying a universal co-$T$ condition is precisely captured by the existence of a so-called odd co-$T$ ranking on $\scrG_{w}$. Complementation then reduces to recognition of $\Delta$-graphs that admit odd co-$T$ rankings.

The general scheme goes as follows. Vertices of $\scrG_{w}$ are associated with certain values. The association at a level can be viewed as a function with domain $Q$ (with level indices dropped), called \emph{co-$T$ level ranking}. The values in the range of a co-$T$ level ranking are called \emph{co-$T$ ranks}, and the $n$-tuple of co-$T$ ranks at a level is called a \emph{co-$T$ level rank}. By a \emph{co-$T$ ranking} we mean an $\omega$-sequence of co-$T$ level rankings, each of which is associated with a level in $\scrG_{w}$. Co-$T$ rankings are required to satisfy a \emph{local} property, which holds between every two adjacent levels and is \emph{solely} defined with respect to the unit $\Delta$-graph of the two levels. The local property therefore can be enforced in a step-by-step check by the transitions of $\CA$. But the local property itself is not enough to ensure that a co-$T$ condition holds universally. A special kind of co-$T$ ranking, called \emph{odd} co-$T$ ranking, is singled out. A co-$T$ ranking is \emph{odd} if and only if every path visits certain vertices (called \emph{odd} vertices) infinitely many times. This \emph{global} property can be captured by a B\"{u}chi condition, using the Miyano-Hayashi breakpoint technique for universality (alternation) elimination~\cite{MH84}.

Let $\calA = \langle Q, Q_{0}, \Sigma, \Delta, \calF^{T} \rangle$ be a source $T$-automaton. The complementation algorithm produces a target B\"{u}chi automaton $\CA=\langle Q', Q'_{0}, \Sigma, \Delta', \langle F'\rangle \rangle$. The state set $Q'$ is $2^{Q} \times 2^{Q} \times \calR$, where for $\langle S, O, g \rangle \in Q'$, $S$ records the reachable states, $O \subseteq S$ records the reachable states that have an obligation to visit \emph{odd} vertices in the future, and $g$ is a guessed co-$T$ level ranking, all at the current level. The transition function $\Delta': Q' \to 2^{Q'}$ is defined such that $\Delta'(\langle S, O, g\rangle)$ is
\begin{align}
&\{\ \langle \Delta(S, \sigma), \Delta(O,\sigma) \setminus odd(g'), g' \rangle: g' \in \Succ(g, S, \sigma) \ \} & (O & \not = \emptyset), \label{eq:non-empty} \\
&\{\ \langle \Delta(S, \sigma), \Delta(S,\sigma) \setminus odd(g'), g' \rangle: g' \in \Succ(g, S, \sigma) \ \} & (O &= \emptyset), \label{eq:empty}
\end{align}
where $\Succ(g, S, \sigma)$ returns the set of legitimate level rankings provided that the current level ranking is $g$ and the current letter is $\sigma$, and $odd(g')$ gives the set of odd vertices at the level ranked by $g'$. When $\CA$ reads the letter $w(i)$ at level $i$ with a level ranking $f_{i}$ and reachable state set $S_{i}$, it nondeterministically guesses a level ranking $f_{i+1}$ such that $f_{i+1} \in \Succ(f_{i}, S_{i}, w(i))$. The evolvement of both $S$ and $O$ are done by the classic subset construction~\cite{RS59} with the exception that odd vertices are excluded from $O$ (see ``$\setminus odd(g')$'' in~\eqref{eq:non-empty} and~\eqref{eq:empty}). Once $O$ becomes empty, it takes the value of the current $S$ in the next stage (see the second $S$ in~\eqref{eq:empty}). The final state set $F'$ is $2^{Q} \times \{\emptyset\} \times \calR$. This B\"{u}chi condition $\langle F'\rangle$ requires that $O$ be cleared infinitely often, which in turn enforces that every path visits odd vertices infinitely many times~\cite{MH84}. It is now clear that $\Succ$ represents the local property (being co-$T$) and $\calF'$ captures the global property (being odd).

\begin{procedure}[Generic Complementation] \mbox{}\\
\label{pro:generic-complementation}
\!\!Input: $T$-automaton $\calA = \langle \Sigma, Q, Q_{0}, \Delta, \calF^{T} \rangle$. Output: B\"{u}chi automaton $\CA=\langle \Sigma, Q', Q'_{0}, \Delta', \langle F' \rangle \rangle$:
\begin{align*}
Q' &= 2^{Q} \times 2^{Q} \times \calR, &
Q'_{0} &= Q_{0} \times \{ \emptyset \} \times \calR, \\
\Delta' &: Q' \to 2^{Q'} \textit{ defined as in~\eqref{eq:non-empty} and~\eqref{eq:empty}}, &
F' &= 2^{Q} \times \{\emptyset\} \times \calR .
\end{align*}
\end{procedure}
The state complexity of complementation is $|Q'|$. For every instantiation shown below, $|\calR|$ dominates $2^{|Q|}$ and hence the complexity is $O(|\calR|)$.

We now show complementation constructions for B\"{u}chi, $\GB$ and Streett, with the corresponding co-$T$ rankings being co-B\"{u}chi, $\GC$ and Rabin, respectively. We use $\calD^{T}$ to denote the set of $T$-ranks, $\calR^{T}$ the set of $T$ level rankings and $\calL^{T}$ the set of level ranks. Clearly, $|\calR^{T}|=|\calL^{T}|$.

\paragraph{B\"{u}chi Complementation}
Let $\calA = \langle Q, Q_{0}, \Sigma, \Delta, \langle F \rangle \rangle$ be a B\"{u}chi automaton. $\scrG_{w}$ is co-B\"{u}chi accepting if \emph{every} path in $\scrG_{w}$ visits $F$-vertices finitely often. Let $\calD^{\CB}$ (the set of co-B\"{u}chi ranks) be $[2n+1]$.
\begin{definition}[Co-B\"{u}chi Ranking]
\label{def:CB-ranking}
A co-B\"{u}chi ranking on $\scrG_{w}$ is a function $f: V \to \calD^{\CB}$ such that:
\begin{enumerate}[label=\ref{def:CB-ranking}.\arabic*,ref=\ref{def:CB-ranking}.\arabic*]
\item \label{en:CB-ranking-1} for all vertices $v \in V$, if $f(v) \in [2n]^{\odd}$, then $v \not \in F$;
\item \label{en:CB-ranking-2} for all edges $\langle v, v'\rangle \in E$, $f(v) \ge f(v')$.
\end{enumerate}
\end{definition}
A vertex $v \in V$ is \emph{odd} if $f(v) \in [2n]^{\odd}$. A co-B\"{u}chi ranking $f$ is \emph{odd} if every path in $\scrG_{w}$ visits infinitely many odd vertices. A path $\varrho$ stabilizes at a rank $r$ if $(\exists i \in \bbN)(\forall j \ge i), f(\varrho(j))=f(\varrho(i))=r$ and the smallest such $i$ is called the \emph{stabilization point} of $\varrho$. If $\scrG_{w}$ admits an odd co-B\"{u}chi ranking $f$, then by~\eqref{en:CB-ranking-2} every path eventually stabilizes at an odd rank. Then by~\eqref{en:CB-ranking-1}, every path eventually does not visit $F$-vertices; that is, $\scrG_{w}$ is co-B\"{u}chi accepting.

Conversely, if $\scrG_{w}$ is co-B\"{u}chi accepting, then an odd co-B\"{u}chi ranking can be constructed through a series of graph transformations. Let $\scrG_{0} = \scrG_{w}$. Vertices with only a finite number of descendants are called \emph{finite}. Vertices that are not $F$-vertices and have no $F$-vertices as their descendants are called \emph{$F$-free}. At stage $0$, we assign all finite vertices rank $0$ and remove them, obtaining $\scrG_{1}$, in which there is no finite vertices. Because $\scrG_{0}$ is co-B\"{u}chi accepting, there must exist in $\scrG_{1}$ an $F$-free vertex; otherwise we can select a path on which $F$-vertices occur infinitely often. We assign all $F$-free vertices rank $1$ and remove them too, obtaining $\scrG_{2}$. Now some vertices in $\scrG_{2}$ are finite due to the removal of $F$-free vertices in stage $0$. We repeat this process in the following manner: at the first phase of stage $i$, we assign even rank $2i$ to finite vertices and remove them; at the second phase, we assign $F$-free vertices odd rank $2i+1$ and remove them. By $F$-freeness, removing $F$-free vertices from $\scrG_{2i+1}$ gets rid of at least one infinite path, and hence $\width(\scrG_{2i+2}) < \width(\scrG_{2i})$. Therefore, this process terminates at a stage $j \le n$. In the following summary, by $\scrG \setminus V$ we mean removing from $\scrG$ all vertices in $V$ and their incoming and outgoing edges.
\begin{procedure}[Co-B\"{u}chi Ranking Assignment] \mbox{}\\
\label{pro:CB-assignment}
\!\!Input: a co-B\"{u}chi accepting $\scrG_{0}$.
Output: a co-B\"{u}chi ranking $f$.
Repeat for $i \in [0..n]$ if $\scrG_{2i} \not = \emptyset$.
\begin{multicols}{2}
\begin{enumerate}[label=\ref{pro:CB-assignment}.\arabic*,ref=\ref{pro:CB-assignment}.\arabic*]
\item \label{en:CB-surgery-1}
(a) $V_{2i} = \{v \in V \mid v \textit{ is finite in } \scrG_{2i} \}$; \\
(b) $f(v) = 2i$ for $v \in V_{2i}$; \\
(c) $\scrG_{2i+1} = \scrG_{2i} \setminus V_{2i}$.
\item \label{en:CB-surgery-2}
(a) $V_{2i+1} = \{v \in V \mid v \textit{ is $F$-free in } \scrG_{2i+1} \}$; \\
(b) $f(v) = 2i+1$ for $v \in V_{2i+1}$; \\
(c) $\scrG_{2i+2} = \scrG_{2i+1} \setminus V_{2i+1}$.
\end{enumerate}
\end{multicols}
\end{procedure}

\begin{lemma}[\cite{KV01}]
\label{lem:CB-ranking}
$\scrG_{w}$ is co-B\"{u}chi accepting if and only if $\scrG_{w}$ admits an odd co-B\"{u}chi ranking.
\end{lemma}
We have $|\calD^{\CB}|=O(n)$, and hence $|\calR^{\CB}|=(O(n))^{n}=2^{O(n \lg n)}$.

\paragraph{GB Complementation}
Let $\calA = \langle Q, Q_{0}, \Sigma, \Delta, \langle B\rangle_{I} \rangle$ be a generalized B\"{u}chi automaton. $\GC$ ranking is meant to be used for $\GB$ complementation. A $\scrG_{w}$ is $\GC$ accepting if for \emph{every} path $\varrho$ in $\scrG_{w}$ there exists $j \in I$ such that $\varrho$ only visits $B(j)$-vertices finitely often. Let $\calD^{\GC} = ({[2n]}^{\odd} \times I) \cup {[2n+1]}^{\even}$ be the set of $\GC$ ranks. We refer to values in ${[2n]}^{\odd} \times I$ as \emph{odd} ranks, and values in ${[2n+1]}^{\even}$ as \emph{even} ranks. For an odd $\GC$ rank $\langle t, u \rangle$, we call $t$ numeric rank ($r$-rank) and $u$ index rank ($h$-rank). Even $\GC$ ranks are just numeric ranks. The greater-than and less-than orders on $\GC$ ranks are solely defined on $r$-ranks. For example, $\langle t, u \rangle > \langle t', u' \rangle$ (or $\langle t, u \rangle > t'$, $t > \langle t', u' \rangle$) if and only if $t > t'$. This definition is sound with respect to its usage in this paper; as shown below, we never need to compare two odd $\GC$ ranks having the same $r$-rank but different $h$-ranks.
\begin{definition}[$\GC$ Ranking]
\label{def:GC-ranking}
A $\GC$ ranking on $\scrG_{w}$ is a function $f: V \to \calD^{GC}$ such that:
\begin{enumerate}[label=\ref{def:GC-ranking}.\arabic*,ref=\ref{def:GC-ranking}.\arabic*]
\item\label{en:GC-ranking-1} for every vertex $v \in V$, if $f(v) = \langle 2i+1, j\rangle$ for some $j \in I$, then $v \not \in B(j)$;
\item\label{en:GC-ranking-2} for every edge $\langle v, v'\rangle \in E$, $f(v) \ge f(v')$.
\end{enumerate}
\end{definition}
A vertex $v$ is called \emph{odd} (resp. \emph{even}) if $f(v) \in {[2n]}^{\odd} \times I$ (resp. $f(v) \in {[2n+1]}^{\even}$). A $\GC$ ranking $f$ is \emph{odd} if every path in $\scrG_{w}$ visits infinitely many odd vertices. Note that~\eqref{en:GC-ranking-2} implies that if two adjacent odd vertices have the same $r$-rank, then they have the same $h$-rank. As in B\"{u}chi complementation, if $\scrG_{w}$ admits an odd $\GC$ ranking, then every path eventually stabilizes at an odd $\GC$ rank $\langle t, j \rangle$, and from the stabilization point on never visits $B(j)$-vertices. Therefore, $\scrG_{w}$ is $\GC$ accepting.

Conversely, if $\scrG_{w}$ is $\GC$ accepting, then we can find a $\GC$ ranking by a series of graph transformations as in B\"{u}chi complementation. Each stage has two phases. We begin stage $i$ with $\scrG_{2i}$ ($\scrG_{0}=\scrG_{w}$). In the first phase, finite vertices receive even rank $2i$ and are removed, resulting in $\scrG_{2i+1}$. Thanks to $\GC$ condition, if $\scrG_{2i+1}$ is not empty, then for some $j \in I$ and $v \in V$, $v$ is $B(j)$-free. In the second phase, those $B(j)$-free vertices receive odd rank $\langle 2i+1, j\rangle$ and are removed, producing $\scrG_{2i+2}$. This procedure repeats for all $i \in [0..n]$ unless $\scrG_{2i}$ is empty. The termination condition is justified by $\width(\scrG_{2i+2}) < \width(\scrG_{2i})$, just as before.

\begin{procedure}[$\GC$ Ranking Assignment] \mbox{}\\
\label{pro:GC-assignment}
\!\!Input: a $\GC$ accepting $\scrG_{0}$.
Output: a $\GC$ ranking $f$.
Repeat for $i \in [0..n]$ if $\scrG_{2i} \not = \emptyset$.
\begin{multicols}{2}
\begin{enumerate}[label=\ref{pro:GC-assignment}.\arabic*,ref=\ref{pro:GC-assignment}.\arabic*]
\item\label{en:GC-surgery-1}
(a) $V_{2i} = \{v \in V \mid \ v \textit{ is finite in } \scrG_{2i} \}$; \\
(b) $f(v) = 2i$ for $v \in V_{2i}$; \\
(c) $\scrG_{2i+1} = \scrG_{2i} \setminus V_{2i}$.\columnbreak
\item\label{en:GC-surgery-2}
(a) $V_{2i+1} = \{v \in V \mid$ $v$ is $B(j)$-free in $\scrG_{2i+1} \}$ for a $j \in I$ such that $B(j)$-free vertices exist; \\
(b) $f(v) = \langle 2i+1, j\rangle$ for $v \in V_{2i+1}$; \\
(c) $\scrG_{2i+2} = \scrG_{2i+1} \setminus V_{2i+1}$.
\end{enumerate}
\end{multicols}
\end{procedure}
In~\eqref{en:GC-surgery-2}, it does not matter which $j \in I$ is chosen. But this flexibility plays an important role in our Streett complementation construction (see Procedure~\ref{pro:new-GC-assignment}).
\begin{lemma}[\cite{KV04}]
\label{lem:GC-ranking}
$\scrG_{w}$ is $\GC$ accepting if and only if $\scrG_{w}$ admits an odd $\GC$ ranking.
\end{lemma}
We have $|\calD^{\GC}|=O(nk)$, and hence $|\calR^{\GC}|=(O(nk))^{n}=2^{O(n\lg nk)}$.

\paragraph{Streett Complementation.}
Let $\calA=\langle Q, Q_{0}, \Sigma, \Delta, \langle G, B \rangle_{I} \rangle$ be a Streett automaton. Rabin ranking is meant for Streett complementation. Let us first examine the simple case where $k=1$, i.e., every path satisfies $[G(1),B(1)]$. Easily seen, $\scrG_{w}$ admits a co-B\"{u}chi ranking, and hence we can instantiate Procedure~\ref{pro:generic-complementation} with $\calR$ being co-B\"{u}chi level rankings (which are also $\GC$ level rankings with index size $1$). The only modification needed is to enforce that every path visits $G(1)$-vertices, which can be easily realized by a B\"{u}chi accepting condition (see the definition of $\langle F' \rangle$ in Procedure~\ref{pro:generic-complementation}). This simple procedure fails for $k > 1$, because a path visiting a finite number of $B(j)$-vertices may not have to visit infinitely many $G(j)$-vertices; it just satisfies $[G(j'),B(j')]$ for $j' \not = j$. Nevertheless, if we could find a way to reduce the number of Rabin pairs one by one, eventually the simple scenario has to occur. The idea in~\cite{KV05a} is to use $\GC$ rankings to approximate Rabin accepting behaviors step by step until finally obtaining the precise characterization. As a result, Rabin ranks are tuples of $\GC$ ranks, considerably more sophisticated than $\GC$ ranks. We first put aside the formal definition of Rabin rankings and show how a Rabin ranking can be obtained provided $\scrG_{w}$ is Rabin accepting. Once again, this is done through a series of graph transformations.
\begin{procedure}[Rabin Ranking Assignment] \mbox{}\\
\label{pro:Rabin-assignment}
\!\!Input: a Rabin accepting $\scrG_{0}$.
Output: a Rabin ranking $f$.
Repeat for $i \in [0..k]$ if $\scrG_{i} \not = \emptyset$.
\begin{enumerate}[label=\ref{pro:Rabin-assignment}.\arabic*,ref=\ref{pro:Rabin-assignment}.\arabic*]
\item\label{en:Rabin-surgery-1} Assign $\scrG_{i}$ a $\GC$ ranking $\gc_{i+1}$.
\item\label{en:Rabin-surgery-2} Remove all vertices $v$ if $\gc_{i+1}(v)$ is even.
\item\label{en:Rabin-surgery-3} Remove all edges $\langle v, v' \rangle$ if $\gc_{i+1}(v) > \gc_{i+1}(v')$.
\item\label{en:Rabin-surgery-4} Remove all edges $\langle v, v' \rangle$ if $\gc_{i+1}(v)$ is odd with index $j$ and $v$ is a $G(j)$-vertex.
\item\label{en:Rabin-surgery-5} $f(v)=\langle \gc_{1}(v), \ldots \gc_{i+1}(v) \rangle$ iff $v$ is removed from $\scrG_{i}$.
\end{enumerate}
\end{procedure}
Obviously, if $\scrG_{w}$ is Rabin accepting for a Rabin condition $[G,B]_{I}$, then it is also $\GC$ accepting for the $\GC$ condition $[B]_{I}$. By Lemma~\ref{lem:GC-ranking}, a $\GC$ ranking $\gc_{1}$ exists for $\scrG_{0}$, which justifies Step~\eqref{en:Rabin-surgery-1} at stage $0$. Steps~\eqref{en:Rabin-surgery-2}-\eqref{en:Rabin-surgery-3} may break up $\scrG_{0}$ into a collection of graph components (in the undirected sense). Let $\calC$ be such a component. Steps~\eqref{en:Rabin-surgery-2}-\eqref{en:Rabin-surgery-3} ensure that vertices in $\calC$ have the same odd rank with some index $j \in I$, and hence all are $B(j)$-free. Step~\eqref{en:Rabin-surgery-4}, deleting all outgoing edges from $G(j)$-vertices, may further break up $\calC$ into more components. In particular, any infinite path is destroyed (i.e., broken into a collection of finite paths) if the path satisfies $[G(j), B(j)]$ (i.e., visiting infinitely many $G(j)$-vertices but only finitely many $B(j)$-vertices). Let $\calC' \subseteq \calC$ be a resulting component after Step~\eqref{en:Rabin-surgery-4}. As a result, $\calC'$ should satisfy the \emph{reduced} Rabin condition $[G, B]_{I \setminus \{j\}}$, and hence the \emph{reduced} $\GC$ condition $[B]_{I \setminus \{j\}}$. So after stage $0$, $\scrG_{1}$ is composed of a collection of pairwise disjoint components, each of which satisfies a Rabin condition whose cardinality is at most $k-1$. Precisely speaking, at the beginning of each stage $i \ge 1$, $\scrG_{i}$ is composed of a collection of pairwise disjoint components, and at Step~\eqref{en:Rabin-surgery-1}, $\gc_{i+1}$ is obtained by \emph{independently} assigning each component in $\scrG_{i}$ a $\GC$ ranking according to the \emph{reduced} $\GC$ condition the component satisfies. By induction, at stage $i \ge 1$, vertices in each component in $\scrG_{i}$ have been assigned the same tuple of odd $\GC$ ranks of length $i$ and each component satisfies a Rabin condition whose cardinality is at most $k-i$.  It follows that the procedure terminates and each vertex in $\scrG_{w}$ eventually gets a tuple of $\GC$ ranks of length at most $k+1$. Note that the last $\GC$ rank in a tuple is always an even $\GC$ rank ($r$-rank).

Let $\calD^{\Rabin}$ denote the set of Rabin ranks of the form $\langle \langle r_{1}, i_{1} \rangle, \ldots, \langle r_{m}, i_{m} \rangle, r_{m+1} \rangle$ ($m \le k$). Ordering relations ($<_{m}, \le_{m}, >_{m}, \ge_{m}, =_{m}$) on Rabin ranks are defined to be the standard lexicographical extension (up to $m$-th component) of orderings on $\GC$ ranks. For a Rabin rank $\gamma$ of the above form, the \emph{index projection} (or the \emph{$h$-projection}) of $\gamma$, written $\Projh \gamma$, is $\langle i_{1}, \ldots, i_{m} \rangle$ and the \emph{numeric projection} (or the \emph{$r$-projection}) of $\gamma$, written $\Projr \gamma$, is $\langle r_{1}, \ldots, r_{m+1} \rangle$. With respect to a given function $f: V \to \calD^{\Rabin}$, the \emph{width} of $v \in V$ is the length of $f(v)$, denoted by $|v|_{f}$ (or $|v|$, when $f$ is clear from the context). We say that $v$ is \emph{odd} (called \emph{happy} in~\cite{KV05a}) if $|v|>1$ and $v$ is a $G(|v|-1)$-vertex. We arrive at the formal definition of Rabin rankings.
\begin{definition}[Rabin Ranking]
\label{def:Rabin-ranking}
A Rabin ranking is a function $f: V \to \calD^{\Rabin}$ satisfying the following conditions.
\begin{enumerate}[label=\ref{def:Rabin-ranking}.\arabic*,ref=\ref{def:Rabin-ranking}.\arabic*]
\item\label{en:Rabin-ranking-1} For every vertex $v \in V$ with $|v|=m+1 \ge 2$ and $\alpha=\Projh f(v)$, we have
\begin{multicols}{2}
\begin{enumerate}
\item\label{en:Rabin-ranking-1a} for $i \in [1..m)$, $v \not \in G(\alpha[i])$;
\item\label{en:Rabin-ranking-1b} for $i \in [1..m]$, $v \not \in B(\alpha[i])$.
\end{enumerate}
\end{multicols}
\item \label{en:Rabin-ranking-2} For every edge $\langle v, v' \rangle \in E$ with $|v|=m+1$, $|v'|=m'+1$ and $m''=\min(m, m')$, we have
\begin{multicols}{2}
\begin{enumerate}
\item \label{en:Rabin-ranking-2a} $f(v) \ge_{m''} f(v')$;
\item \label{en:Rabin-ranking-2b} $f(v) \ge_{m''+1} f(v')$, or $v$ is odd.
\end{enumerate}
\end{multicols}
\end{enumerate}
\end{definition}
A Rabin ranking is \emph{odd} if every path in $\scrG_{w}$ visits infinitely many odd vertices.

\begin{lemma}[\cite{KV05a}]
\label{lem:Rabin-ranking}
$\scrG_{w}$ is Rabin accepting if and only if $\scrG_{w}$ admits an odd Rabin ranking.
\end{lemma}
We have $|\calD^{\Rabin}|=(O(nk))^{k+1}$ and hence $|\calR^{\Rabin}|=((O(nk))^{k+1})^{n}=(nk)^{O(nk)}=2^{O(nk \lg nk)}$. 
\section{Improved Streett Complementation}
\label{sec:streett-complementation}

The above construction requires $2^{O(nk \lg nk)}$ state blow-up~\cite{KV05a}, which is substantially larger than the lower bound in~\cite{CZ11a}. In the extreme case of $k=O(2^{n})$, the construction is double exponential in $n$. Intuitively, the larger the $k$, the more overlaps between $B(i)$'s and between $G(i)$'s ($i \in I$). A natural question is: can all Rabin pairs $[G(i), B(i)]$ independently impose behaviors on a Rabin accepting $\scrG_{w}$? We showed in~\cite{CZL09} that in Rabin complementation we can build a Streett accepting $\scrG_{w}$ for which no Streett pair $\langle G(i), B(i)\rangle$ is redundant. We observed the opposite in Streett complementation; the larger the $k$, the higher the correlation between infinite paths that satisfy $[G,B]_{I}$. By exploiting this correlation, we can walk in \emph{big} steps in approximating Rabin accepting behaviors using $\GC$ rankings. As a result, our Rabin ranks are tuples of $\GC$ ranks of length at most $\mu=\min(n,k)$. This simple but crucial observation leads to a significant improvement on the construction complexity. We elaborate on this below.

The first idea is that at Step~\eqref{en:Rabin-surgery-4} in stage $i$, in a component $\calC$ of $\scrG_{i}$, instead of removing all outgoing edges from $G(j)$-vertices, we can remove all outgoing edges from $G(j')$-vertices for all $j'$ such that $B(j') \subseteq B(j)$. Let $J=\{j' \in I \mid B(j') \subseteq B(j)\}$. Since vertices in $\calC$ are $B(j)$-free, they are also $B(j')$-free for any $j' \in J$. Recall that Step~\eqref{en:Rabin-surgery-4} is to break all infinite paths that satisfy $[G(j),B(j)]$ so that in each resulting component we have a simpler Rabin condition to satisfy. If an infinite path in $\calC$ satisfies $[G,B]_{J}$, then removing all outgoing edges from $G(j')$-vertices (for $j' \in J$) certainly serves the same purpose, and moreover, any resulting component only needs to satisfy a Rabin condition whose cardinality is $|J|$ less than at the beginning of stage $i$.

The second idea is that at Step~\eqref{en:Rabin-surgery-1} in stage $i$, we can assign special $\GC$ rankings to components in $\scrG_{i}$. Recall that a $\GC$ ranking is obtained by a series of graph transformations too. In Step~\eqref{en:GC-surgery-2}, we assign and remove $B(j)$-free vertices for some $j \in I$. In fact any fixed $j \in I$ is sufficient as long as $B(j)$-free vertices exist. Therefore, we can choose a $j$ such that not only $B(j)$-free vertices exist, but also for any other $j' \in I$, $B(j') \not \subset B(j)$, if $B(j')$-free vertices also exist. Intuitively, we prefer a $j$ such that $B(j)$ is minimal (with respect to set inclusion) because more vertices would be $B(j)$-free and subject to removal.

We refine those ideas by taking into account the history of $\GC$ rankings. In stage $i$, right before Step~\eqref{en:Rabin-surgery-1}, vertices in $\scrG_{i}$ were assigned a tuple of $\GC$ ranks of length $i$. Consider a component $\calC \subseteq \scrG_{i}$. No vertices in $\calC$ received an even $\GC$ rank in stage $i' \in [0..i)$, because otherwise they were already removed by Step~\eqref{en:Rabin-surgery-2} in that stage. Also all vertices in $\scrG_{i}$ received the same odd $\GC$ rank in each stage $i' \in [0..i)$, for otherwise Step~\eqref{en:Rabin-surgery-3} in stage $i'$ would have broken the component. Now let $\langle \langle r_{1}, j_{1} \rangle, \ldots, \langle r_{i}, j_{i} \rangle \rangle$ be the tuple that has been assigned to all vertices in $\calC$. Let $B'=\cup_{t \in [1..i]} B(j_{t})$ and $J'=\{j' \in I \mid B(j') \subseteq B' \}$. So all vertices in $\calC$ are $B'$-free. When we assign $\GC$ rankings for $\calC$, in each stage at Step~\eqref{en:GC-surgery-2} (in Procedure~\ref{pro:GC-assignment}), we choose a $j \in I \setminus J'$ such that (1) $B(j)$-free vertices exist, (2) $B(j) \not \subseteq B'$ (we say $B(j)$ is not \emph{covered} by $B'$), and (3) no $B(j')$-free vertices exist for any other $j' \in I \setminus J'$ with $B' \cup B(j') \subset B' \cup B(j)$. In other words, we choose a $j$ such that not only we can find $B(j)$-free vertices, but also $B' \cup B(j)$ minimally extends $B'$. Now let $\calC'$ be a component right before Step~\eqref{en:Rabin-surgery-4} is taken and let $\langle \langle r_{1}, j_{1} \rangle, \ldots, \langle r_{i+1}, j_{i+1} \rangle \rangle$ be the tuple of ranks that has been assigned to all vertices in $\calC'$ (for the same reason as before, all vertices in $\calC'$ have received the same sequence of odd $\GC$ ranks). Let $B''=B' \cup B(j_{i+1})$ and $J''=\{j'' \in I \mid B(j'') \subseteq B'' \}$. In Step~\eqref{en:Rabin-surgery-4} we remove all outgoing edges of $G(j'')$-vertices if $B(j'') \subseteq B''$. This deletion destroys all infinite paths that satisfy $[G,B]_{J''}$ because all vertices in $\calC'$ are $B''$-free. Let $\calC''$ be a resulting component and $\varrho$ an infinite path in $\calC''$. Then $\varrho$ only needs to satisfy $[G,B]_{I \setminus J''}$.

Now let us assume that we have incorporated the above ideas into Procedure~\ref{pro:Rabin-assignment} and obtained a new Rabin rank $\gamma = \langle \langle r_{1}, j_{1} \rangle, \ldots, \langle r_{m}, j_{m} \rangle, r_{m+1} \rangle \rangle$. Let $\alpha=\Projh \gamma$. The \emph{non-covering} condition stated above requires $\alpha$ to satisfy:
\begin{align}
\label{eq:non-covering}
\forall i \in [1..m]\ B(\alpha[i]) \not \subseteq \cup_{j=1}^{i-1} B(\alpha[j]) ,
\end{align}
which implies $|\alpha| \le n$. By definition, $|\alpha| \le k$ and so we have $|\alpha| \le \mu$ and $|\gamma| \le \mu+1$. From now on we switch to terms \emph{$\muR$ ranks} (i.e., \emph{minimal Rabin ranks}), \emph{$\muR$ rankings}, \emph{$\muR$ level rankings} and \emph{$\muR$ level ranks}. Their precise definitions are to be given below. We define two functions $\Cover: I^{*} \to 2^{I}$ and $\Mini: I^{*} \to 2^{I}$ to formalize the intuition of \emph{minimal extension}. $\Cover$ maps tuples of indices to subsets of $I$ such that
\begin{align*}
\Cover(\alpha) = \{\, j \in I \ \mid \ B(j) \subseteq \cup_{i=1}^{|\alpha|} B(\alpha[i]) \, \} .
\end{align*}
Note that $\Cover(\epsilon)=\emptyset$. $\Mini$ maps tuples of indices to subsets of $I$ such that $j \in \Mini(\alpha)$ if and only if $j \in I \setminus \Cover(\alpha)$ and
\begin{align}
& \forall j' \in I \setminus \Cover(\alpha) \, \big(j' \not = j \to
 B(j') \cup \Cover(\alpha) \not \subset B(j) \cup \Cover(\alpha) \big) \, , \label{eq:mini-1} \\
& \forall j' \in I \setminus \Cover(\alpha) \, \big(j' < j \to
 B(j') \cup \Cover(\alpha) \not = B(j) \cup \Cover(\alpha) \big) \, .\label{eq:mini-2}
\end{align}
$\Mini(\alpha)$ consists of choices of indices to \emph{minimally} enlarge $\Cover(\alpha)$; ties (with respect to set inclusion) are broken by numeric minimality (Condition~\eqref{eq:mini-2}). Before introducing $\muR$ ranking assignment, we need a new $\GC$ ranking assignment which takes a tuple of $I$-indices as an additional input. We call so obtained $\GC$ rankings (resp. $\GC$ ranks) $\muGC$ rankings (resp. $\muGC$ ranks).

\begin{procedure}[$\muGC$ Ranking Assignment] \mbox{}\\
\label{pro:new-GC-assignment}
\!\!Input: a $\GC$ accepting $\scrG_{0}$, a tuple of $I$-indices $\alpha$.
Output: a $\muGC$ ranking $f$. \\
Repeat for $i \in [0..n]$ if $\scrG_{2i} \not = \emptyset$.
\setlength{\columnsep}{0pt}
\begin{multicols}{2}
\begin{enumerate}[label=\ref{pro:new-GC-assignment}.\arabic*,ref=\ref{pro:new-GC-assignment}.\arabic*]
\item\label{en:new-GC-surgery-1}
(a) $V_{2i} = \{v \in V \mid \ v \textit{ is finite in } \scrG_{2i} \}$; \\
(b) $f(v) = 2i$ for $v \in V_{2i}$; \\
(c) $\scrG_{2i+1} = \scrG_{2i} \setminus V_{2i}$.\columnbreak
\item\label{en:new-GC-surgery-2}
(a) $V_{2i+1} = \{v \in V \mid$ $v$ is $B(j)$-free in $\scrG_{2i+1} \}$ for a $j \in \Mini(\alpha)$ such that $B(j)$-free vertices exist; \\
(b) $f(v) = \langle 2i+1,j \rangle$ for $v \in V_{2i+1}$; \\
(c) $\scrG_{2i+2} = \scrG_{2i+1} \setminus V_{2i+1}$.
\end{enumerate}
\end{multicols}
\end{procedure}
As in $\GC$ ranking assignment, in~\eqref{en:new-GC-surgery-2} there maybe more than one $j$ such that $B(j)$-free vertices exist, and it does not matter which one we choose. But in case $\Mini(\alpha)$ is a singleton, we have a unique $j$ at all stages, essentially synchronizing all $h$-ranks in the $\muGC$ ranks obtained. This synchronization is crucial in our construction for parity complementation (see Section~\ref{sec:parity-complementation}).

\begin{procedure}[$\muR$ Ranking Assignment] \mbox{}\\
\label{pro:new-Rabin-assignment}
\!\!Input: a Rabin accepting $\scrG_{0}$.
Output: a $\muR$ ranking $f$.
Repeat for $i \in [0..\mu]$ if $\scrG_{i} \not = \emptyset$.
\begin{enumerate}[label=\ref{pro:new-Rabin-assignment}.\arabic*,ref=\ref{pro:new-Rabin-assignment}.\arabic*]
\item\label{en:new-Rabin-surgery-1} Assign $\scrG_{i}$ a $\muGC$ ranking $\gc_{i+1}$.
\item\label{en:new-Rabin-surgery-2} Remove all vertices $v \in V$ if $\gc_{i+1}(v)$ is even.
\item\label{en:new-Rabin-surgery-3} Remove all edges $\langle v, v' \rangle \in E$ if $\gc_{i+1}(v) > \gc_{i+1}(v')$.
\item\label{en:new-Rabin-surgery-4} Remove all edges $\langle v, v' \rangle \in E$ if $v \in G(t)$ for some
$t \in \Cover(\Projh(\langle \gc_{1}, \ldots, \gc_{i+1}\rangle))$.
\item\label{en:new-Rabin-surgery-5} $f(v)=\langle \gc_{1}(v), \ldots \gc_{i+1}(v) \rangle$ iff $v$ is removed from $\scrG_{i}$.
\end{enumerate}
\end{procedure}
Note that Step~\eqref{en:new-Rabin-surgery-1} actually means that Procedure~\ref{pro:new-GC-assignment} is called upon for every component $\calC \subset \scrG_{i}$, with the corresponding $\alpha$ being $\Projh (\langle \gc_{1}(v), \ldots, \gc_{i}(v)\rangle)$ for some $v \in \calC$ ($\alpha$ is well-defined since all vertices in $\calC$ have received the same sequence of $\muGC$ ranks). It is time to formally define $\muR$ ranking. Let $f$ be a function $V \to (\calD^{\GC})^{\mu+1}$. We say that $v$ is \emph{odd} if $|v|>1$ and $v$ is a $G(t)$-vertex for some $t \in \Cover(\alpha[1..|v|-1])$ where $\alpha=\Projh f(v)$.
\begin{definition}[$\muR$ Ranking]
\label{def:new-Rabin-ranking}
A $\muR$ ranking is a function $f: V \to (\calD^{\GC})^{\mu+1}$ satisfying the following conditions.
\begin{enumerate}[label=\ref{def:new-Rabin-ranking}.\arabic*,ref=\ref{def:new-Rabin-ranking}.\arabic*]
\item\label{en:new-Rabin-ranking-1} For every vertex $v \in V$ with $|v|=m+1 \ge 2$ and $\alpha=\Projh f(v)$, we have
\begin{enumerate}
\item\label{en:new-Rabin-ranking-1a} for $i \in [1..m)$, $v \not \in G(t)$ for $t \in \Cover(\alpha[1..i])$;
\item\label{en:new-Rabin-ranking-1b} for $i \in [1..m]$, $v \not \in B(t)$ for $t \in \Cover(\alpha[1..i])$;
\item\label{en:new-Rabin-ranking-1c} for $i \in [1..m]$, $\alpha[i] \in \Mini(\alpha[1..i))$.
\end{enumerate}
\item \label{en:new-Rabin-ranking-2} For every edge $\langle v, v' \rangle \in E$ with $|v|=m+1$, $|v'|=m'+1$ and $m''=\min(m, m')$, we have
\begin{multicols}{2}
\begin{enumerate}
\item \label{en:new-Rabin-ranking-2a} $f(v) \ge_{m''} f(v')$;
\item \label{en:new-Rabin-ranking-2b} $f(v) \ge_{m''+1} f(v')$, or $v$ is odd.
\end{enumerate}
\end{multicols}
\end{enumerate}
\end{definition}
%%%%%%%%%%%%%%%%%%%%%%%%%%%%%%%%%%%%%%%%%%%%%%%%
A $\muR$ ranking is \emph{odd} if every infinite path in $\scrG_{w}$ visits infinitely many odd vertices.
\begin{lemma}
\label{lem:new-Rabin-ranking}
$\scrG_{w}$ is Rabin accepting if and only if $\scrG_{w}$ admits an odd $\muR$ ranking.
\end{lemma} 
\section{Complexity}
\label{sec:complexity}

In this section we analyze the complexity of our construction. As shown below, all complexity related notions $X$ are also parameterized with $B$ (besides $n$ and $k$). Still, we choose to list some or all of them when clarity is needed. In particular, we use $|X|$ to abbreviate $|X(n,k)|=\max_{B} |X(B,n,k)|$.

Let $\calD^{\muR}$ be the set of $\muR$ ranks that can be produced by Procedure~\ref{pro:new-Rabin-assignment}. Formally, $\calD^{\muR} = \cup_{f} \range(f)$ where $f$ ranges over all possible outputs of Procedure~\ref{pro:new-Rabin-assignment}. In a similar manner, we define $\calR^{\muR}$ and $\calL^{\muR}$. We have $|\calR^{\muR}| = |\calL^{\muR}|$. We note that these notions are defined differently from their counterparts in Section~\ref{sec:ranking-based-complementation}; due to two kinds of correlations (which we refer to as \emph{horizontal} and \emph{vertical} correlation), $|\calL^{\muR}|$ is much smaller than $|\calD^{\muGC}|^{n(\mu+1)}$. We view $\calL^{\muR}$ as a set of $n \times (\mu+1)$ matrices of $\muGC$ ranks and carry out a further simplification. Let $\calM^{\muR}$ be a set of $n \times \mu$ matrices obtained from $\calL^{\muR}$ by the following mapping: each $n \times (\mu+1)$ matrix $M$ is mapped to an $n \times \mu$ matrix $M'$ by (a) deleting from $M$ even ranks at the end of each row, (b) changing odds rank of the form $\langle 2i-1, j \rangle$ to $\langle i, j \rangle$, and (c) aligning each row to length $\mu$ by filling $\langle 1, 0 \rangle$'s. Clearly, $|\calL^{\muR}| \le n^{n} \cdot |\calM^{\muR}|$; the factor $n^{n}$ suffices to compensate (a), the deletion of even ranks, and (b) and (c) have no effect to the cardinality (for (b), $i \to 2i-1$ is one-to-one from $[1..n]$ onto $[2n]^{\odd}$). Let $M \in \calM^{\muR}$ be called $\muR$-matrices. We write $\Projr M$ and $\Projh M$ to mean, respectively, the projection of $M$ on numeric ranks (called an $r$-matrix) and on index ranks (called an $h$-matrix). Let $\calM^{r}=\Projr \calM^{\muR}$, $\calM^{h}=\Projh \calM^{\muR}$, and $\calD^{r}$ and $\calD^{h}$ be the sets of rows occurring in matrices in $\calM^{r}$ and $\calM^{h}$, respectively. Obviously, we have $|\calM^{\muR}| \le |\calM^{r}| \cdot |\calM^{h}|$ and $|\calM^{h}| \le (|\calD^{h}|)^{n}$.

\begin{example}[$\muR$-Matrix]
\label{ex:muR-matrix}
Let us consider a case where $n=3$, $k=3$, and $Q=\{q_{0}, q_{1}, q_{2}\}$. Below we show that a $\muR$ level rank $f$
corresponds to a $\muR$-matrix $M$, which projects to $M_{r}$ and $M_{h}$.
\begin{align*}
\begin{array}{ccccc}
\begin{bmatrix}
q_{0} \\
q_{1} \\
q_{2} \\
\end{bmatrix}
&
\begin{vmatrix}
\langle 1, 2 \rangle & \langle 1, 3 \rangle & 4 &   \\
\langle 1, 2 \rangle & \langle 3, 1 \rangle & 2 &   \\
\langle 1, 2 \rangle & \langle 3, 1 \rangle & \langle 3, 3 \rangle & 0
\end{vmatrix}
&
\begin{vmatrix}
\langle 1, 2 \rangle & \langle 1, 3 \rangle & \langle 1, 0 \rangle \\
\langle 1, 2 \rangle & \langle 2, 1 \rangle & \langle 1, 0 \rangle \\
\langle 1, 2 \rangle & \langle 2, 1 \rangle & \langle 2, 3 \rangle
\end{vmatrix}
&
\begin{vmatrix}
1 & 1 & 1 \\
1 & 2 & 1 \\
1 & 2 & 2
\end{vmatrix}
&
\begin{vmatrix}
2 & 3 & 0 \\
2 & 1 & 0 \\
2 & 1 & 3
\end{vmatrix} \\
Q & f & M & M_{r} & M_{h}
\end{array}
\end{align*}
\end{example}

\paragraph{Bounding $|\calM^{h}|$.}
It turns out that we only need to exploit a horizontal correlation to bound $|\calM^{h}|$. Recall that each $\alpha \in I^{*}$ names a subset of $Q$, namely $\Cover(\alpha)$. The idea is to order all $\alpha$ that could occur in $\calD^{h}$ into a tree structure. Consider an unordered tree where the root is labeled by $\epsilon$ and each non-root node is labeled by an index in $I$. With little confusion, we identify a node $\alpha$ with the path from the root to $\alpha$ and represent $\alpha$ by the sequence of indices on the path. So a non-root node $\alpha$ has $\alpha[|\alpha|]$ as its label and names $\Cover(\alpha)$. We arrive at the following important notion.

\begin{definition}[Increasing Tree of Sets ($\ITS$)]
\label{def:ITS} An $\ITS$ $T(n,k,B)$ is an unordered $I$-labeled tree (except the root which is labeled by $\epsilon$) such that
\begin{enumerate}[label=\ref{def:ITS}.\arabic*,ref=\ref{def:ITS}.\arabic*]
\item \label{en:ITS} A non-root node $\alpha$ exists in $T(n,k,B)$ iff $\forall i \in [1..|\alpha|]$, $\alpha[i] \in \Mini(\alpha[1..i))$.
\end{enumerate}
\end{definition}

Property~\eqref{en:ITS} succinctly encodes three important features of $\ITS$. First, an $\ITS$ is maximal in the sense that no node can be added. Second, if $\beta$ is a direct child of $\alpha$, then $\beta$ must name at least one new state that has not been named by $\alpha$. Third, the new contributions by $\beta$ cannot be covered by contributions made by any another sibling $\beta'$. In particular, if more than one sibling can make the same contribution, then the one with the smallest index is selected. It follows that each tuple of $n$, $k$ and $B$ uniquely determines $T(n,k,B)$ (in the unordered sense). Note that the height of $T(n,k,B)$ (the length of the longest path in $T(n,k,B)$) is bounded by $\mu$.

\begin{example}[$\ITS$]
\label{ex:ITS}
Consider $n=3$, $k=4$, $Q=\{q_{0},q_{1},q_{2}\}$, $B:[1..4] \to 2^{Q}$ and $B':[1..5] \to 2^{Q}$,
\begin{align*}
B(1)=\{q_{0}, q_{1}\}, && B(2)=\{q_{0}\}, && B(3)=\{q_{1},q_{2}\}, && B(4)=\{q_{2}\},
\end{align*}
and $B'$ extends $B$ with $B'(5)=\{q_{1}\}$. $T(3,4,B)$ and $T(3,5,B')$ are given in Figure~\ref{fig:ITS}.
For clarity, for each non-root node $\alpha$, we also list $B(\alpha[|\alpha|])$ as the set label of $\alpha$. In $T(3,4,B)$, neither $\{q_{0},q_{1}\}$ nor $\{q_{1},q_{2}\}$ can appear at height $1$, because $\{q_{0},q_{1}\}$ covers $\{q_{0}\}$ and $\{q_{1},q_{2}\}$ covers $\{q_{2}\}$. The leftmost node at the bottom level is labeled by $\{q_{1},q_{2}\}$ instead of by $\{q_{2}\}$ due to the index minimality requirement. For the same reason, in $T(3,5,B')$, we have nodes $\langle 2, 1\rangle$, $\langle 2,4,1\rangle$ and $\langle 4,2,1\rangle$ all labeled with $\{q_{0},q_{1}\}$, and nodes $\langle 2, 1, 3 \rangle$ and $\langle 4,3\rangle$ all labeled with $\{q_{1},q_{2}\}$.
\end{example}

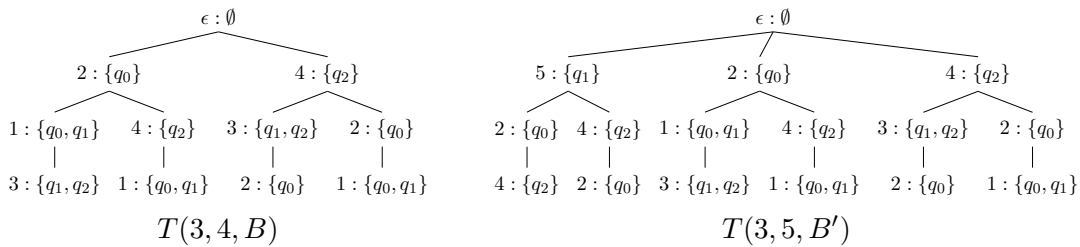
\begin{figure}[t!]
\begin{center}
\begin{align*}
\begin{array}{ccc}
\begin{tikzpicture}[scale=0.7]
\Tree [.$\epsilon:\emptyset$
        [.$2:\{q_{0}\}$
          [.$1:\{q_{0},q_{1}\}$
            [.$3:\{q_{1},q_{2}\}$ ]
          ]
          [.$4:\{q_{2}\}$
            [.$1:\{q_{0},q_{1}\}$ ]
          ]
        ]
        [.$4:\{q_{2}\}$
          [.$3:\{q_{1},q_{2}\}$
            [.$2:\{q_{0}\}$ ]
          ]
          [.$2:\{q_{0}\}$
            [.$1:\{q_{0},q_{1}\}$ ]
          ]
        ]
      ]
\end{tikzpicture}
& &
\begin{tikzpicture}[scale=0.7]
\Tree [.$\epsilon:\emptyset$
        [.$5:\{q_{1}\}$
          [.$2:\{q_{0}\}$
            [.$4:\{q_{2}\}$ ]
          ]
          [.$4:\{q_{2}\}$
            [.$2:\{q_{0}\}$ ]
          ]
        ]
        [.$2:\{q_{0}\}$
          [.$1:\{q_{0},q_{1}\}$
            [.$3:\{q_{1},q_{2}\}$ ]
          ]
          [.$4:\{q_{2}\}$
            [.$1:\{q_{0},q_{1}\}$ ]
          ]
        ]
        [.$4:\{q_{2}\}$
          [.$3:\{q_{1},q_{2}\}$
            [.$2:\{q_{0}\}$ ]
          ]
          [.$2:\{q_{0}\}$
            [.$1:\{q_{0},q_{1}\}$ ]
          ]
        ]
      ]
\end{tikzpicture} \\
T(3,4,B) & & T(3,5,B')
\end{array}
\end{align*}
\end{center}
\caption{\textrm{Two $\ITS$ in Example~\ref{ex:ITS}.}}
\label{fig:ITS}
\end{figure}

It is easily seen that Property~\eqref{en:new-Rabin-ranking-1c} corresponds exactly to Property~\eqref{en:ITS}. So a one-to-one correspondence exists between non-root nodes in $T(n,k,B)$ and elements in $\calD^{h}(n,k,B)$. Let $|T(n,k,B)|$ denote the number of non-root nodes in $T(n,k,B)$ and $H(n,k)=\max_{B} |T(n,k,B)|$. Clearly, we have $|\calM^{h}| \le (H(n,k))^{n}$.

\begin{lemma}
\label{lem:bound-of-TH}
$H(n,k) = 2^{O(k \lg k)}$ for $k = O(n)$ and $H(n,k) = 2^{O(n \lg n)}$ for $k = \omega(n)$.
\end{lemma}

\paragraph{Bounding $|\calM^{r}|$.}
Here we need to exploit both horizontal and vertical correlations. We show that every $n \times \mu$ $r$-matrix induces a $2^{Q}$-labeled ordered tree with at most $n$ leaves and with height at most $\mu$. Such a tree is called an $n \times \mu$ tree. We bound $|\calM^{r}|$ by counting the number of $n \times \mu$ trees.

Let $M$ be an $r$-matrix. Since $M$ comes from a $\muR$ level rank, $M$ is associated with vertices at a level. To facilitate the discussion below, we use term \emph{states} to specifically mean those vertices at the level where $M$ is associated with, and use term \emph{vertices} just as before. By rank $i$ we simply mean a number $i$ in $M$, which corresponds to the numeric $\muGC$ rank $2i-1$.

Let us first consider ranks in column $1$ of $M$. A state $q$ being ranked with an odd $\muGC$ rank means that at certain stage of $\muGC$ ranking assignment, $q$ becomes $B(j)$-free for some $j \in I$, which implies that there exists an infinite path starting from $q$ (recall that all finite vertices have been removed before this odd $\muGC$ rank is assigned). If two states $q$, $q'$ are ranked with different odd $\muGC$ ranks, say $q$ with $\langle 2i-1,j\rangle$ and $q'$ with $\langle 2i'-1,j'\rangle$ where $i > i'$, then there exist two infinite paths $\varrho$ and $\varrho'$ such that $\varrho$ starts from $q$, $\varrho'$ starts from $q'$, and $\varrho$ and $\varrho'$ never intersect. This is due to the nature of $\muGC$ ranking assignment; $\langle 2i-1,j\rangle$ is assigned to some $B(j)$-free vertices only after those $B(j')$-free vertices with odd $\muGC$ rank $\langle 2i'-1,j'\rangle$ have been removed.

Note that it is perfectly possible that an infinite path starting from $q$ intersects another infinite path starting from $q'$. But a maximal subset $S^{(1)}$ of states, all with the same rank, called a \emph{cell} at column $1$, should ``own'' at least one \emph{private} infinite path that does not intersect the private paths owned by any other cells at column $1$. We call a path \emph{named} if it is owned by a cell. Let $m^{(1)}$ be the maximum rank in column $1$, and note that not all ranks in $[1..m^{(1)}]$ necessarily appear in column $1$. But again, by the way $\muGC$ ranking assignment is carried out, for each non-occurring rank, at least one private infinite path exists, which is called \emph{hidden} and viewed as being owned by $\emptyset$. Easily seen now, each rank in $[1..m^{(1)}]$ corresponds to a non-empty set of private infinite paths.

In general, a \emph{cell} at column $l$ is a maximal subset of states, each of which is assigned the same tuple of ranks up to column $l$. Consider a cell $S^{(l)}=\{q_{i_{1}}, \ldots, q_{i_{j}}\}$ at column $l$. Let $m^{(l+1)}=\max\{M[i_{1},l+1], \ldots, M[i_{j},l+1]\}$. By the same reasoning as before, each rank in $[1..m^{(l+1)}]$ corresponds to a non-empty set of private infinite paths. The private paths associated with rank $M[i_{j'},l+1]$ are owned by $S^{(l+1)}_{j'} \subseteq S^{(l)}$ which is a cell at column $l+1$ with rank $M[i_{j'},l+1]$ ($S^{(l+1)}_{j'}=\emptyset$ if no states in $S^{(l)}$ is mapped to $M[i_{j'},l+1]$). Moreover, none of these paths, hidden or named, should intersect private paths owned by any cell at column $l$ that is a subset of $Q \setminus S^{(l)}$, because states in $Q \setminus S^{(l)}$ and states in $S^{(l)}$ are not in the same component at stage $l$ (the $l+1$-th stage) in Procedure~\ref{pro:new-Rabin-assignment}. Now we are ready to show how to build an $n \times \mu$ tree from $M$.

Each node in the tree is associated with a label which is a subset of $Q$. The root is labeled with set $Q$. For each rank $i \in [1..m^{(1)}]$, we add a child to the root and we order those children increasingly by the ranks associated with them. If $i$ does not appear in column $1$, the $i$-th child (from left to right) is labeled with $\emptyset$ and is a terminal node (leaf). Otherwise, the child is labeled with the cell at column $1$ with rank $i$ and the child is non-terminal if its height is less than $\mu$. We repeat the process column by column. Each maximal $S^{(l)}$ at column $l < \mu$ corresponds to a non-terminal node at height $l$, from which we spawn a child for each rank in $i \in [1..m^{(l)}]$, and we order and label the children using the rule stated above. After processing column $\mu$, we obtain an $n \times \mu$ tree, because the number of leaves in the tree cannot exceed $\width(\scrG_{w})\le n$, which we refer to as the \emph{maximum width property} ($\MWP$). Now we call an $n \times \mu$ tree a $\TOP$ (\emph{Tree of Ordered Partitions}) and let $\calT^{r}(n,k)$ denote the set of $\TOPs$. We have $|\calM^{r}| \le |\calT^{r}(n,k)|$.

\begin{figure}[t!]
\begin{center}
\begin{tabular}{cccc}
$
\begin{vmatrix}
1 & 1 & 2 \\
1 & 1 & 1 \\
1 & 1 & 1
\end{vmatrix}
$ &
$
\begin{vmatrix}
1 & 1 & 3 \\
1 & 1 & 1 \\
1 & 1 & 1
\end{vmatrix}
$ &
$
\begin{vmatrix}
  2 & 1 & 2 \\
  2 & 1 & 2 \\
  1 & 2 & 1
\end{vmatrix}
$ &
$
\begin{vmatrix}
2 & 1 & 2 \\
2 & 1 & 2 \\
1 & 2 & 3
\end{vmatrix}
$ \\
$M_{1}$ & $M_{2}$ & $M_{3}$ & $M_{4}$ \\
\begin{tikzpicture}[scale=0.7]
\Tree [.$\{q_{0},q_{1},q_{2}\}$
        [.$\{q_{0},q_{1},q_{2}\}$
          [.$\{q_{0},q_{1},q_{2}\}$
            [.$\{q_{1},q_{2}\}$ ]
            [.$\{q_{0}\}$ ]
          ]
        ]
      ]
\end{tikzpicture}
&
\begin{tikzpicture}[scale=0.7]
\Tree [.$\{q_{0},q_{1},q_{2}\}$
        [.$\{q_{0},q_{1},q_{2}\}$
          [.$\{q_{0},q_{1},q_{2}\}$
            [.$\{q_{1},q_{2}\}$ ]
            \edge[dashed];
            [.$\emptyset$ ]
            [.$\{q_{0}\}$ ]
          ]
        ]
      ]
\end{tikzpicture}
&
\begin{tikzpicture}[scale=0.7]
\Tree [.$\{q_{0},q_{1},q_{2}\}$
        [.$\{q_{2}\}$
          [.$\{q_{2}\}$
            [.$\{q_{2}\}$ ]
          ]
        ]
        [.$\{q_{0},q_{1}\}$
          [.$\{q_{0},q_{1}\}$
            \edge[dashed];
            [.$\emptyset$ ]
            [.$\{q_{0},q_{1}\}$ ]
          ]
        ]
      ]
\end{tikzpicture}
&
\begin{tikzpicture}[scale=0.7]
\Tree [.$\{q_{0},q_{1},q_{2}\}$
        [.$\{q_{2}\}$
          \edge[dashed];
          [.$\emptyset$ ]
          [.$\{q_{2}\}$
            [.$\{q_{2}\}$ ]
          ]
        ]
        [.$\{q_{0},q_{1}\}$
          [.$\{q_{0},q_{1}\}$
            \edge[dashed];
            [.$\emptyset$ ]
            [.$\{q_{0},q_{1}\}$ ]
          ]
        ]
      ]
\end{tikzpicture} \\
$T_{1}$ & $T_{2}$ & $T_{3}$ & $T_{4}$
\end{tabular}
\end{center}
\caption{\textrm{Four $3 \times 3$ matrices and their corresponding tree representations in Example~\ref{ex:TOP}.}}
\label{fig:TOP}
\end{figure}
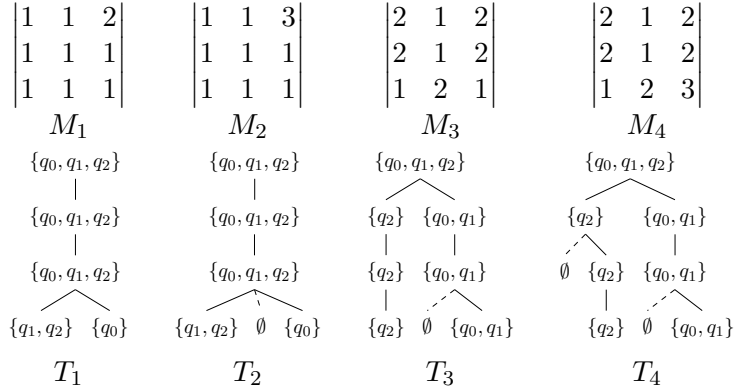

\begin{example}[$\TOP$]
\label{ex:TOP}
Four $3 \times 3$ matrices $M_{1}$-$M_{4}$ and their corresponding tree representations $T_{1}$-$T_{4}$ are given in Figure~\ref{fig:TOP}. $M_{1}$-$M_{3}$ obey $\MWP$ and hence $T_{1}$-$T{3}$ are $\TOPs$. $T_{4}$ is not a $\TOP$ because it has more than $3$ leaves.
\end{example}

\begin{lemma}[Numeric Bound]
\label{lem:numeric-bound}
$|\calT^{r}(n,k)| = 2^{O(n \lg n)}$.
\end{lemma}

Since $|\calR^{\muR}| \le n^{n} \cdot |\calM^{\muR}|$, $|\calM^{\muR}| \le |\calM^{r}| \cdot |\calM^{h}|$, $|\calM^{r}| \le |\calT^{r}(n,k)|$, and $|\calM^{h}| \le (H(n,k))^{n}$, by Lemmas~\ref{lem:bound-of-TH} and~\ref{lem:numeric-bound}, we have

\begin{theorem}[Streett Upper Bound]
\label{thm:Streett-upper-bound}
Streett complementation is $2^{O(n \lg n + nk \lg k)}$ for $k = O(n)$ and $2^{O(n^{2} \lg n)}$ for $k=\omega(n)$.
\end{theorem}
Note that we put bounds in the form $2^{O(\cdot)}$ just for simplicity. Even for a small $k$ (i.e. $k=O(n)$), our upper bound is substantially smaller than the current best one $(nk)^{O(nk)}$, established by respective constructions in~\cite{Kla91,Saf92,KV05a,Pit06}. Easily seen from the proofs of Lemmas~\ref{lem:bound-of-TH} and ~\ref{lem:numeric-bound}, our upper bound is in fact $n^{O(n)} \cdot k^{O(nk)}$ when $k=O(n)$.
Also note that Lemma~\ref{lem:numeric-bound} is crucial in tightening parity complementation. 
\section{Parity Complementation}
\label{sec:parity-complementation}

Parity automata is a special kind of Streett automata where a Streett condition $\langle G, B\rangle_{I}$ is augmented with the so-called \emph{Rabin chain condition} $B(1) \subset G(1) \subset \cdots \subset B(k) \subset G(k)$. Now the short length of $\muR$ ranks is not enough to give us a better bound, because we already have $k \le \lfloor (n+1)/2\rfloor$. Nevertheless, the Rabin chain condition makes the $\GC$ condition $[B]_{I}$ degenerate to the $\CB$ condition $[B(1)]$, because being $B(1)$-free is equivalent to being $B(i)$-free for some $i \in I$. This coincides with the way $\Mini$ works. Intuitively, $\Mini$ synchronizes all components at a stage of $\muR$ ranking assignment. In the first stage of $\muR$ ranking assignment, in Step~\eqref{en:Rabin-surgery-1}, $\Mini$ makes every vertex get the $h$-rank $1$, though vertices may get different $r$-ranks. After disabling $G(1)$ vertices (by deleting all outgoing edges from them in Step~\eqref{en:Rabin-surgery-4}), we have a collection of components satisfying parity condition $\langle G, B\rangle_{[2..k]}$. Then in the second stage of $\muR$ ranking assignment, $\Mini$ gives every vertex the $h$-rank $2$. Repeating this process, the $h$-projection of a final $\muR$ rank is just $\langle 1, \ldots, m \rangle$ for some $m \in [1..k]$, which is completely redundant, because the only useful information (having length $m$) is already encoded by the corresponding $r$-projection. As a consequence, $h$-matrices contribute nothing to the complexity. A customized construction for parity complementation is given in the appendix.

\begin{theorem}[Parity Upper Bound]
\label{thm:parity-upper-bound}
Parity complementation is in $2^{O(n\lg n)}$.
\end{theorem}
This bound matches the lower bound of B\"{u}chi complementation, and hence it is tight as B\"{u}chi automata are a subclass of parity automata. To the best of our knowledge, the previous best upper bound is $2^{O(nk \lg n)}$, which can be easily inferred from~\cite{KV05a} by treating parity automata as Rabin automata. 
\section{Concluding Remarks}
\label{sec:conclusion}
In this paper we improved Kupferman and Vardi's construction and obtained tight upper bounds for Streett and parity complementation (with respect to the $2^{\Theta(X)}$ asymptotic notation). Figure~\ref{fig:bound-summary} in the appendix rounds up the complementation complexities for $\omega$-automata of common types.

Our inquiry also leads to some unexpected outcomes, which we believe, would help understand the strength and weakness of different types of $\omega$-automata in modeling and specifying system behaviors.
\begin{enumerate}
\item Parity complementation has the same asymptotical bound as B\"{u}chi complementation while parity automata have richer and more elegant acceptance conditions than B\"{u}chi automata.
\item Streett automata are exponentially more succinct than B\"{u}chi automata while Rabin automata are not. On the other hand, Streett complementation is much easier than Rabin complementation when $k$ is large (i.e., $k=\omega(n)$). In the extreme case where $k=\Theta(2^{n})$ and $N=\Theta(nk)$ (the automata size), Streett complementation is in $O(N^{\lg^{2} N})=O(2^{\lg^{3} N})$ while Rabin complementation is still in $2^{\Omega(N)}$.
\end{enumerate}

Further investigation on Streett and parity complementation is desired as exponential gaps can hide in the asymptotical notations of the form $2^{\Theta(X)}$. The situation is different from that of B\"{u}chi where the best lower and upper bounds have been shown polynomially close.

We think that $\ITS$ and $\TOP$ characterize intrinsic combinatorial properties on run graphs with universal Rabin conditions. Interesting questions remain for further investigation. What would be the counterparts for run graphs with existential Streett conditions? The discovery of such combinatorial properties might help us understand the complexity of Streett determinization, for which there exists a huge gap between the current lower bound $2^{\Omega(n^{2} \lg n)}$~\cite{CZ11a} and upper bound $2^{O(nk\lg nk)}$~\cite{Pit06} when $k=\omega(n)$. Also of theoretical interest is whether there exists a type of $\omega$-automata whose determinization is considerably harder than complementation. In the case of B\"{u}chi, the two operations were both proved to be in $2^{\Theta(n \lg n)}$.

\newpage
\appendix
\section{Parity Complementation Construction}

Parity ranking is used for parity complementation. Let $\scrG_{0}$ be a parity accepting $\Delta$-graph. The following procedure assign a parity ranking to $\scrG_{0}$.
\begin{procedure}[Parity Ranking Assignment]\mbox{}\\
\label{pro:parity-assignment}
\!\!Input: a parity accepting $\scrG_{0}$.
Output: a parity ranking $f$.
Repeat for $i \in [0..k]$ if $\scrG_{i} \not = \emptyset$.
\begin{enumerate}[label=\ref{pro:parity-assignment}.\arabic*,ref=\ref{pro:parity-assignment}.\arabic*]
\item\label{en:parity-surgery-1} Assign $\scrG_{i}$ a co-B\"{u}chi ranking $\cb_{i+1}$ with respect to the co-B\"{u}chi condition $[B(i+1)]$.
\item\label{en:parity-surgery-2} Remove all vertices $v \in V$ if $\cb_{i+1}(v)$ is even.
\item\label{en:parity-surgery-3} Remove all edges $\langle v, v' \rangle \in E$ if $\cb_{i+1}(v) > \cb_{i+1}(v')$.
\item\label{en:parity-surgery-4} Remove all edges $\langle v, v' \rangle \in E$ if $v \in G(i+1)$.
\item\label{en:parity-surgery-5} $f(v)=\langle \cb_{1}(v), \ldots \cb_{i+1}(v) \rangle$ iff $v$ is removed from $\scrG_{i}$.
\end{enumerate}
\end{procedure}

Let $\calD^{\Parity}$ denote the set of parity ranks, a set of tuples of co-B\"{u}chi ranks of length at most $k+1$, which can be produced by Procedure~\ref{pro:parity-assignment}. Similar as before, let $|v|$ be the width of $v$ with respect to a given function $f: V \to (\calD^{\CB})^{k+1}$.  We say that $v$ is \emph{odd} if $|v|>1$ and $v \in G(|v|-1)$.
\begin{definition}[Parity Ranking]
\label{def:parity-ranking}
A parity ranking is a function $f: V \to (\calD^{\CB})^{k+1}$ satisfying the following conditions.
\begin{enumerate}[label=\ref{def:parity-ranking}.\arabic*,ref=\ref{def:parity-ranking}.\arabic*]
\item\label{en:parity-ranking-1} For every vertex $v \in V$ with $|v|=m+1 \ge 2$, we have $v \not \in B(m)$.
\item\label{en:parity-ranking-2} For every edge $\langle v, v' \rangle \in E$ with $|v|=m+1$, $|v'|=m'+1$ and $m''=\min(m, m')$, we have
\begin{enumerate}
\item \label{en:parity-ranking-2a} $f(v) \ge_{m''} f(v')$.
\item \label{en:parity-ranking-2b} $f(v) \ge_{m''+1} f(v')$, or $v$ is odd.
\end{enumerate}
\end{enumerate}
\end{definition}
A parity ranking is \emph{odd} with respect to $\scrG_{w}$ if every infinite path in $\scrG_{w}$ visits infinitely many odd vertices.

\begin{lemma}
\label{lem:parity-ranking}
$\scrG_{w}$ is parity accepting if and only if $\scrG_{w}$ admits a parity ranking.
\end{lemma}

\section{Proofs}
\label{sec:proof}

We split Lemma~\ref{lem:new-Rabin-ranking} to Lemmas~\ref{lem:new-Rabin-ranking-1} and~\ref{lem:new-Rabin-ranking-2}. Lemma~\ref{lem:new-Rabin-ranking-1} requires the following two additional lemmas.

\begin{lemma}
\label{lem:first-nonincreasing}
Let $\varrho$ be an infinite path in $\scrG$ that admits a $\muR$ ranking $f$. Then for any $i \in \bbN$, $f(\varrho(i)) \ge_{1} f(\varrho(i+1))$.
\end{lemma}
\begin{proof}
Let $i \in \bbN$, $|\varrho(i)|=m+1$, $|\varrho(i+1)|=m'+1$ and $m''=\min(m,m')$. By Property~\eqref{en:new-Rabin-ranking-2b}, $f(\varrho(i)) \ge_{m''+1} f(\varrho(i+1))$ unless $\varrho(i)$ is odd. If $\varrho(i)$ is not odd, then $f(\varrho(i)) \ge_{m''+1} f(\varrho(i+1))$ implies $f(\varrho(i))[1] \ge f(\varrho(i+1))[1]$. But if $\varrho(i)$ is odd, then by definition $m'' \ge 1$. Now by Property~\eqref{en:new-Rabin-ranking-2a}, we still have $f(\varrho(i))[1] \ge f(\varrho(i+1))[1]$, that is $f(\varrho(i)) \ge_{1} f(\varrho(i+1))$.
\end{proof}

\begin{lemma}
\label{lem:finitely-even}
If $\scrG$ admits a $\muR$ ranking $f$, then any infinite path in $\scrG$ has only finitely many vertices of width $1$.
\end{lemma}

\begin{proof}
Let $\varrho$ be an infinite path that contains infinitely many vertices of width $1$. By Lemma~\ref{lem:first-nonincreasing} we have
\begin{align}
\label{eq:nonincreasing-1}
f(\varrho(0)) \ge_{1} f(\varrho(1)) \ge_{1} f(\varrho(2)) \ge_{1} \ldots
\end{align}
Because $f$ is a $\muR$ ranking for $\scrG$, $\varrho$ has infinitely many odd vertices. Since odd vertices have width greater than $1$, $\varrho$ has to contain infinitely many edges between vertices of width $>1$ and vertices of width $1$. Then in the above sequence, infinitely many relations are $>_{1}$, a contradiction to well-foundedness.
\end{proof}

\begin{lemma}
\label{lem:new-Rabin-ranking-1}
If $\scrG$ admits a $\muR$ ranking $f$, then all paths of $\scrG$ satisfy $[G,B]_{I}$.
\end{lemma}

\begin{proof}
Let $f$ be a $\muR$ ranking for $\scrG$ and $\varrho$ an infinite path in $\scrG$. By Lemma~\ref{lem:finitely-even}, $\varrho$ contains a suffix in which all vertices have width $>1$. Let $m''$ be such that $m''+1$ is the minimum width of vertices in $\varrho$ that appears infinitely often. So $m'' \ge 1$.

We show that there exist infinitely many odd vertices with width $m''+1$. Suppose the opposite. Since there are infinitely many odd vertices and $m''+1$ is the minimum width of vertices occurring infinitely often, from some point on in $\varrho$, all odd vertices (there are infinitely many) have width $> m''+1$. By Property~\eqref{en:new-Rabin-ranking-2b}, for some $j_{0} \ge 0$, we have an infinite non-increasing sequence
\begin{align}
\label{eq:nonincreasing-2}
f(\varrho(j_{0})) \ge_{m''+1} f(\varrho(j_{0}+1)) \ge_{m''+1} f(\varrho(j_{0}+2)) \ge_{m''+1} \cdots \, .
\end{align}
However, the projection of $m''+1$-th positions of the sequence contains infinitely many even $\GC$ ranks (from vertices of width $m''+1$) as well as odd $\GC$ ranks (from vertices of width $> m''+1$), which together with~\eqref{eq:nonincreasing-3}, contradicts well-foundedness.

By Property~\eqref{en:new-Rabin-ranking-2a}, for some $j_{1} \ge 0$, we have an infinite non-increasing sequence
\begin{align}
\label{eq:nonincreasing-3}
f(\varrho(j_{1})) \ge_{m''} f(\varrho(j_{1}+1)) \ge_{m''} f(\varrho(j_{1}+2)) \ge_{m''} \cdots
\end{align}
which further implies that from some $j_{2} \ge j_{1}$, we have
\begin{align}
\label{eq:nonincreasing-4}
f(\varrho(j_{2})) =_{m''} f(\varrho(j_{2}+1)) =_{m''} f(\varrho(j_{2}+2)) =_{m''} \cdots \, .
\end{align}
By Property~\eqref{en:new-Rabin-ranking-1a}, for any vertex $v$ after level $j_{2}$, $v \not \in B(t)$ for any $t \in \Cover(\alpha[1..m''])$ where $\alpha=\Projh f(v)$. Because there exists infinitely many $j_{3} > j_{2}$ such that $\varrho(j_{3})$ is odd and $|\varrho(j_{3})|=m''+1$, by definition of oddness, we have infinitely many vertex $v'$ such that $v' \in G(t)$ for some $t \in \Cover(\alpha[1..m''])$. Because $\Cover(\alpha[1..m''])$ is finite, there must be some $t'$ such that $G(t')$ is visited infinitely often by $\varrho$. Thus, $\varrho$ satisfies $[G,B]_{I}$, in particular, the condition $[G(t'), B(t')]$.
\end{proof}

\begin{lemma}
\label{lem:new-Rabin-ranking-2}
Let $\scrG$ be a $\Delta$-graph that satisfy $[G,B]_{I}$. Then $\scrG$ admits a $\muR$ ranking.
\end{lemma}

\begin{proof}
Let $f$ be the function produced by Procedure~\ref{pro:new-Rabin-assignment}. We show that all properties in Definition~\ref{def:new-Rabin-ranking} are satisfied by this $f$ and any infinite path in $\scrG_{w}$ visits odd vertices infinitely often.

Property~\eqref{en:new-Rabin-ranking-1}.
Let $v \in V$ be a vertex with width $m+1$ and $f(v)=\langle \langle r_{1}, i_{1} \rangle \ldots, \langle r_{m}, i_{m} \rangle, r_{m+1} \rangle$. Let $\alpha=\Projh f(v)=\langle i_{1}, \ldots, i_{m} \rangle$. We have $m \ge 1$.

Property~\eqref{en:new-Rabin-ranking-1b}.
The fact that $v$ get odd $\GC$ ranks in stage $0, \ldots, m-1$ means that $v \not \in B(\alpha_{j})$ for any $j \in [1..m]$. This implies for any $j \in [1..m]$, $v \not \in B(t)$ for any $t \in \Cover(\alpha[1..j])$.

Property~\eqref{en:new-Rabin-ranking-1a}.
If $m=1$, then this property holds trivially. Suppose $m \ge 2$, and for some $j \in [1..m-1]$, $v \in G(t)$ for some $t \in \Cover(\alpha[1..j])$. Due to Property~\eqref{en:new-Rabin-ranking-1b}, we have $v \not \in B(t)$. Then Step~\eqref{en:new-Rabin-surgery-4} in stage $j-1$ (note that stage numbering starts at $0$ and index numbering starts at $1$) will remove all outgoing edges of $v$, rendering $v$ as a finite vertex in $\scrG_{j}$. So in stage $j$, $v$ will get an even $\GC$ ranks. Since $j \in [1..m-1]$, for some $j' \in [2..m]$, $\alpha[j']$ is an even $\GC$ rank, contradicting the definition of $\alpha$. Therefore, Property~\eqref{en:new-Rabin-ranking-1a} follows.

Property~\eqref{en:new-Rabin-ranking-1c}.
At stage $i$ (for $i < m$), $v$ obtains a $\GC$ ranking with respect to $[G,B]_{I\setminus J}$ where $J = \Cover(\alpha[1..i])$. Therefore, $\alpha_{i+1} \not \in J$ and $B(\alpha_{i+1}) \not \subseteq \cup_{t \in J} B(t)$, where $\cup_{t \in J} B(t)$ is equal to $\cup_{j=1}^{i} B(\alpha[j])=\Cover(\alpha[1..i])$. This says that $\Mini(\alpha[1..i]) \not = \emptyset$. The property then follows from the way we carry out minimal $\GC$ ranking (Step~\eqref{en:new-GC-surgery-2} in Procedure~\ref{pro:new-GC-assignment}).

Property~\eqref{en:new-Rabin-ranking-2}. Let $v, v' \in V$ and $\langle v, v' \rangle \in E$. Let $m, m'$ be such that $|v|=m+1$, $|v'|=m'+1$ and $m'' = \min(m,m')$. Property~\eqref{en:new-Rabin-ranking-2a} trivially holds for $m''=1$. Property~\eqref{en:new-Rabin-ranking-2b} holds for $m''=1$ due to Lemma~\ref{lem:first-nonincreasing}. Now assume that $m''>1$. Let $\theta$ be the maximum number such that $\langle v, v' \rangle$ exists in the same component in $\scrG_{\theta}$. By Definition~\ref{def:GC-ranking}, we have $f(v)[i] \ge f(v')[i]$ for all $i \in [1..\theta+1]$.

Case 1: $\theta=m''$. In this case Property~\eqref{en:new-Rabin-ranking-2} is immediate, because by Property~\eqref{en:GC-ranking-2}, $f(v)[i] \ge f(v')[i]$ (for any $i \in [1..m''+1]$), which implies $f(v) \ge_{m''+1} f(v')$.

Case 2: $\theta<m''-1$. Then $v$ and $v'$ are in different components in $\scrG_{\theta+1}$. So $\langle v, v' \rangle$ must be removed at stage $\theta$. Since $\theta < m''$, both $f(v)[\theta+1]$ and $f(v')[\theta+1]$ are odd, and therefore $\langle v, v' \rangle$ cannot be removed by Step~\eqref{en:new-Rabin-surgery-2}. Nor can it be removed by Step~\eqref{en:new-Rabin-surgery-4}, because otherwise $v$ becomes finite in $\scrG_{\theta+1}$ and removed at stage $\theta+1 < m''$, and hence $|f(v)|=\theta+2 \le m'' < m''+1$, a contradiction. The only possibility left is that $\langle v, v' \rangle$ is removed by Step~\eqref{en:new-Rabin-surgery-3}. Then we have $f(v)[\theta+1] > f(v')[\theta+1]$. We already have $f(v)[i] \ge f(v')[i]$ for all $i \in [1..\theta+1]$, and $\theta+1 < m''$. Therefore, $f(v) >_{m''} f(v')$. We are done with Property~\eqref{en:new-Rabin-ranking-2}.

Case 3: $\theta=m''-1$. In this case, Property~\eqref{en:new-Rabin-ranking-2a} is immediate, because $f(v)[i] \ge f(v')[i]$ ($i \in [1..m'']$) implies $f(v) \ge_{m''} f(v')$. Property~\eqref{en:new-Rabin-ranking-2b} follows immediately if $f(v) >_{m''} f(v')$. So we assume $f(v) =_{m''} f(v')$. As before, $\langle v, v' \rangle$ cannot be removed by Step~\eqref{en:new-Rabin-surgery-2} because both $f(v)[\theta]$ and $f(v')[\theta]$ are odd. But now $\langle v, v' \rangle$ cannot be removed by Step~\eqref{en:new-Rabin-surgery-3}, for otherwise we have $f(v)[\theta+1] > f(v')[\theta+1]$, contradicting the assumption $f(v) =_{m''} f(v')$. Therefore, $\langle v, v' \rangle$ has to be removed by Step~\eqref{en:new-Rabin-surgery-4}, which implies that $|v|=m''+1$ and $v \in G(t)$ for some $t \in \Cover(\alpha[1..m''])$ (where $\alpha=\Projh f(v)$), that is, $v$ is odd.

What is left is to show that any infinite path in $\scrG_{w}$ visits odd vertices infinitely often. Suppose the opposite. Then there must exist an infinite path that visits no odd vertices. Let $\varrho$ be such a path. Let $m$ be such that $m+1$ is the minimum width of vertices that appears infinitely often in $\varrho$. We have $m \ge 1$ because of Lemma~\ref{lem:finitely-even}. Also $\varrho$ must have a suffix in which all vertices have width no less than $m+1$. By Property~\eqref{en:new-Rabin-ranking-2b}, for some $j \ge 0$, we have \begin{align*}
f(\varrho(j)) \ge_{m+1} f(\varrho(j+1)) \ge_{m+1} f(\varrho(j+2)) \ge_{m+1} \cdots \, ,
\end{align*}
and hence for some $j' \ge j$, we have
\begin{align*}
f(\varrho(j')) =_{m+1} f(\varrho(j'+1)) =_{m+1} f(\varrho(j'+2)) =_{m+1} \cdots \, .
\end{align*}
Since vertices with width $m+1$ appears infinitely often in $\varrho$, $f(\varrho(j'))[m+1]$ must be an even $\GC$ rank, which means there is no vertex with width greater than $m+1$ from $\varrho(j')$ on, and for any vertex $j'' \ge j'$, $f(\varrho(j''))[m+1]$ is an even $\GC$ rank. Let $\varrho'$ denote this suffix starting from $\varrho(j')$. We claim that there must be infinitely many $i$ such that edges $\langle \varrho'(i), \varrho'(i+1) \rangle$ do not exist in the same component $\calC$ in $\scrG_{m}$. Suppose otherwise, let $\varrho''$ be the suffix of $\varrho'$ such that for any $i \ge 0$, $\langle \varrho''(i), \varrho''(i+1) \rangle$ appear in the same component in $\scrG_{m}$. Let $f[m+1]$ denote the $(m+1)$-th projection of $f$. We have
\begin{align*}
f[m+1](\varrho''(0)) = f[m+1](\varrho''(1)) = f[m+1](\varrho''(2)) = \cdots \, .
\end{align*}
Although $f[m+1]$ in general may not be a $\GC$ ranking for $\scrG_{m}$, $f[m+1]$, when restricted to a $\calC'$ component in $\scrG_{m}$, is indeed a $\GC$ ranking for $\calC'$. So if all $\varrho''(i)$ ($i \ge 0$) are in the same component $\calC$, then Lemma~\ref{lem:GC-ranking} is violated because $\varrho''$ is a path on which all vertices have even $\GC$ ranks.

Now let us assume that for infinitely many $i$, $\langle \varrho'(i), \varrho'(i+1) \rangle$ is removed at stage $m_{1}$ for some $m_{1} < m$. We have three cases to analyze.

Case 1: For infinitely many $i$, $\langle \varrho'(i), \varrho'(i+1) \rangle$ is removed at stage $m_{1}$ by Step~\eqref{en:new-Rabin-surgery-2}. This is impossible because both $f(\varrho'(i))[m_{1}]$ and $f(\varrho'(i+1))[m_{1}]$ are odd.

Case 2: For infinitely many $i$, $\langle \varrho'(i), \varrho'(i+1) \rangle$ is removed at stage $m_{1}$ by Step~\eqref{en:new-Rabin-surgery-3}. This is also impossible as $f(\varrho'(i))[m_{1}] > f(\varrho'(i+1))[m_{1}]$ contradicts $f(\varrho'(i)) =_{m} f(\varrho'(i+1))$.

Case 3: For infinitely many $i$, $\langle \varrho'(i), \varrho'(i+1) \rangle$ is removed at stage $m_{1}$ by Step~\eqref{en:new-Rabin-surgery-4}. If $m_{1} < m-1$. Then for some $i^{*}$ in those infinitely many $i$'s, $|\varrho'(i^{*})| = m_{1}+1 < m$, contradicting the assumption that all vertices in $\varrho'$ have width $m+1$. So $m_{1} = m-1$. The removal of $\langle \varrho'(i), \varrho'(i+1) \rangle$ by Step~\eqref{en:new-Rabin-surgery-4} is due to $\varrho'(i) \in G(t)$ for some $t \in \Cover(\alpha[1..m])$ where $\alpha=\Projh f(\varrho'(i))$. Recall that we already have $m \ge 1$ due to Lemma~\ref{lem:finitely-even}. They together just say that $\varrho'(i)$ is odd. Because we have infinitely many such $i$, we have infinitely many odd vertices in $\varrho'$, and therefore in $\varrho$.
\end{proof}

\begin{varlem}{\ref{lem:bound-of-TH}.}
$H(n,k) = 2^{O(k \lg k)}$ for $k = O(n)$ and $H(n,k) = 2^{O(n \lg n)}$ for $k = \omega(n)$.
\end{varlem}
\begin{proof}
Recall that for fixed $n$, $k$ and $B: I \to 2^{Q}$, $T(n,k,B)$ is uniquely determined. Also note that the height of $T(n,k,B)$ is bounded by $\mu=\min(n,k)$ and the maximum branching factor is bounded by $k$. We have two cases to consider.
\begin{enumerate}
\item $k = O(n)$. In this case we have $\mu = O(k)$. Therefore, we have
\begin{align*}
|T(n,k,B)| \le \sum_{i=1}^{\mu} k^{i} \le \mu k^{\mu} = 2^{\lg \mu + \mu \lg k} = 2^{O(k \lg k)} .
\end{align*}
Since $B$ is chosen arbitrarily, we have $H(n,k) = 2^{O(k \lg k)}$.
\item $k = \omega(n)$. Assume $k \ge n \ge 2$. Let $B': I' \to 2^{Q}$ where $I'=[1..k']$ be an extension of $B$ such that $\range(B')$ contains all singletons from $Q$. Formally,
\begin{align*}
& \forall i \in [1..k] \ B'(i) = B(i) \, ,
& \forall q \in Q      \ \{q\} \in \range(B') \, .
\end{align*}
By Lemma~\ref{lem:reshuffle-labels}, we assume without loss of generality that $B'(i)=\{q_{i-1}\}$ for $i \in [1..n]$. Due to existence of all singletons, the minimal extension at each node is always done by adding singletons, that is, for any index sequence $\alpha$, $\Mini(\alpha) \subseteq [1..n]$. Therefore, each nonempty path in $T(n,k',B')$ corresponds to a nonempty prefix of a permutation of $[1..n]$ and vice versa. All leaves of $T(n,k',B')$ are at height $n$, and $T(n,k',B')$ has exactly $n!$ leaves and exactly $n!$ internal nodes at height $n-1$. Also, each node at height $j < n-1$ has at least two children, which implies that the total number of nodes at height $j < n-1$ is bounded by $n!$. By Lemma~\ref{lem:add-singleton}, we have
\begin{align*}
|T(n,k,B)| \le |T(n,k',B')| \le 3n! = 2^{O(n \lg n)}.
\end{align*}
As $B$ is chosen arbitrarily, we have $H(n,k) = 2^{O(n \lg n)}$. \hfill \qedhere
\end{enumerate}
\end{proof}

\begin{lemma}
\label{lem:reshuffle-labels}
Let $B:I \to 2^{Q}$, $B': I \to 2^{Q}$ be two injective functions such that $\range(B)=\range(B')$. Then $|T(n,k,B)| = |T(n,k,B')|$.
\end{lemma}
\begin{proof}
The condition $\range(B)=\range(B')$ means that $B$ and $B'$ just name subsets of $Q$ differently. We extend $B$, $B'$ to functions from $I^{*}$ to $2^{Q}$ such that for $\alpha \in I^{*}$,
\begin{align*}
B(\alpha) &= \bigcup_{i=1}^{|\alpha|} B(\alpha[i]) \,, & B'(\alpha) &= \bigcup_{i=1}^{|\alpha|} B'(\alpha[i]) \,.
\end{align*}
By Definition~\ref{def:ITS}, children of a node $\alpha$ in an $\ITS$ is completely determined by $B(\alpha)$. So a node $\alpha$ in $T(n,k,B)$ has the same number of children as a node $\alpha'$ in $T(n,k,B')$ if $B(\alpha)=B'(\alpha')$. Moreover, if $\alpha_{1}, \ldots, \alpha_{j}$ are children of $\alpha$ and $\alpha'_{1}, \ldots, \alpha'_{j}$ are children of $\alpha'$, Then $B(\alpha_{1}), \ldots, B(\alpha_{j})$ are just a permutation of $B'(\alpha'_{1}), \ldots, B'(\alpha'_{j})$. By induction on tree height, there is a one-to-one correspondence between nodes in $T(n,k,B)$ and nodes in $T(n,k,B')$. Thus $|\calT(n,k,B)| = |\calT(n,k,B')|$.
\end{proof}

\begin{lemma}
\label{lem:add-singleton}
Let $T(n,k,B)$ and $T(n,k',B')$ be two $\ITS$ such that $B'$ extends $B$ by naming a singleton that $B$ does not name. Then $|T(n,k,B)| \le |T(n,k',B')|$.
\end{lemma}
\begin{proof}
It suffices to show that $|T(n,k,B)| \le |T(n,k',B')|$ when $k'=k+1$, $I'=[1..k']$, $B'(i)=B(i)$ for $i \in I$ and $B'(k+1)$ names a new singleton. Without loss of generality we assume $B'(k+1)=\{q_{0}\}$. We show a tree transformation $\Theta$ that turns $T(n,k,B)$ into $T(n,k+1,B')$.

Recall that labels on the path from the root to a node $\alpha$ is in this order: $\epsilon, \alpha[1], \ldots, \alpha[|\alpha|]$. Let $T_{\alpha}$ denote the subtree rooted at $\alpha$. We say that a state $q \in Q$ is named by $\alpha$ if $q \in B(\alpha)=\bigcup_{i=1}^{|\alpha|} B(\alpha[i])$. We define a tree transformation $\theta$ such that $\theta(T_{\alpha})$ is as follows.
\begin{enumerate}
\item $q_{0} \in B(\alpha)$. Then $\theta(T_{\alpha})=T_{\alpha}$.
\item $q_{0} \not \in B(\alpha)$. Let $\alpha_{1}, \ldots, \alpha_{l}$ list all children of $\alpha$, where each of $\alpha_{1}, \ldots, \alpha_{j}$ names $q_{0}$, but none of $\alpha_{j+1}, \ldots, \alpha_{l}$ does. Formally,
\begin{align*}
q_{0} &\in B(\alpha_{i})      &   (i \in [1..j], j \ge 0) \, , \\
q_{0} &\not \in B(\alpha_{i}) &   (i \in [j+1..l], j,l \ge 0) \, .
\end{align*}
\begin{enumerate}
\item $(\exists i \in [1..j])\, B(\alpha) \cup \{q_{0}\} = B(\alpha) \cup B(\alpha_{i})$. Then we must have $j=1$. Let $\theta(T_{\alpha})$ be the tree obtained from $T_{\alpha}$ by replacing the label of $\alpha_{1}$ by $k+1$.
\item $(\forall i \in [1..j])\, B(\alpha) \cup \{q_{0}\} \subset B(\alpha) \cup B(\alpha_{i})$. Let $\theta(T_{\alpha})$ be the tree obtained from $T_{\alpha}$ by the following procedure.
\begin{enumerate}
\item Add to $\alpha$ a new child $\beta$ labeled with $k+1$, i.e., $\beta = \alpha \cdot \langle k+1 \rangle$.
\item Remove subtrees $T_{\alpha_{1}}, \ldots, T_{\alpha_{j}}$ from $\alpha$ and make them children of $\beta$.
\item Add to $\beta$ $l-j$ new leaves labeled with $\alpha_{j+1}[|\alpha_{j+1}|], \ldots, \alpha_{l}[|\alpha_{l}|]$.
\item Grow every new leaf into a full $\ITS$ using $\Mini$ with respect to $n$, $k'$ and $B'$.
\end{enumerate}
\end{enumerate}
\end{enumerate}

In any case, $\theta(T_{\alpha})$ is an $\ITS$ with root labeled with $\alpha[|\alpha|]$. It is clear from the above procedure that $|\theta(T_{\alpha})| \ge |T_{\alpha}|$. Now we define $\Theta$ such that $\Theta(T)$ is the tree obtained from $T$ by applying $\theta$ on $T$ level by level, from top to bottom. Example~\ref{ex:ITS-trans} shows such a transformation. It is not hard to verify that $T(n,k',B')=\Theta(T(n,k,B))$. Therefore, we have $|T(n,k,B)| \le |T(n,k',B')|$.
\end{proof}

\begin{example}
\label{ex:ITS-trans}
Let us revisit Example~\ref{ex:ITS} and see how $T(3,3,B)$ is transformed to $T(3,4,B')$ via $\Theta$. Applying $\theta$ at the root of $T(3,3,B)$ we have $T_{1}$:
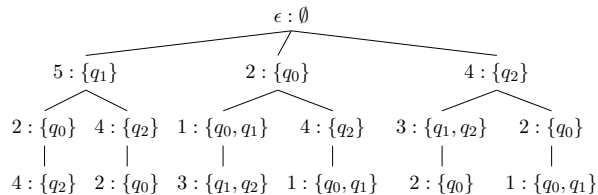
\begin{align*}
\begin{tikzpicture}[scale=0.7]
\Tree [.$\epsilon:\emptyset$
        [.$5:\{q_{1}\}$
          [.$2:\{q_{0}\}$
            [.$4:\{q_{2}\}$ ]
          ]
          [.$4:\{q_{2}\}$
            [.$2:\{q_{0}\}$ ]
          ]
        ]
        [.$2:\{q_{0}\}$
          [.$1:\{q_{0},q_{1}\}$
            [.$3:\{q_{1},q_{2}\}$ ]
          ]
          [.$4:\{q_{2}\}$
            [.$1:\{q_{0},q_{1}\}$ ]
          ]
        ]
        [.$4:\{q_{2}\}$
          [.$3:\{q_{1},q_{2}\}$
            [.$2:\{q_{0}\}$ ]
          ]
          [.$2:\{q_{0}\}$
            [.$1:\{q_{0},q_{1}\}$ ]
          ]
        ]
      ]
\end{tikzpicture}
\end{align*}
Applying $\theta$ at nodes $\langle 2 \rangle$ and $\langle 4 \rangle$ in $T_{1}$, we have $T_{2}$:
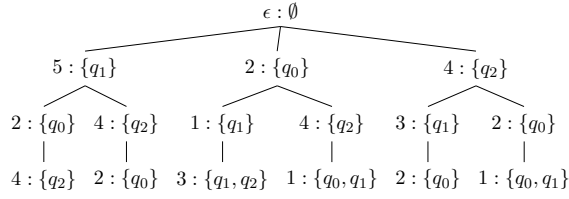
\begin{align*}
\begin{tikzpicture}[scale=0.7]
\Tree [.$\epsilon:\emptyset$
        [.$5:\{q_{1}\}$
          [.$2:\{q_{0}\}$
            [.$4:\{q_{2}\}$ ]
          ]
          [.$4:\{q_{2}\}$
            [.$2:\{q_{0}\}$ ]
          ]
        ]
        [.$2:\{q_{0}\}$
          [.$1:\{q_{1}\}$
            [.$3:\{q_{1},q_{2}\}$ ]
          ]
          [.$4:\{q_{2}\}$
            [.$1:\{q_{0},q_{1}\}$ ]
          ]
        ]
        [.$4:\{q_{2}\}$
          [.$3:\{q_{1}\}$
            [.$2:\{q_{0}\}$ ]
          ]
          [.$2:\{q_{0}\}$
            [.$1:\{q_{0},q_{1}\}$ ]
          ]
        ]
      ]
\end{tikzpicture}
\end{align*}
Applying $\theta$ at nodes $\langle 2,4 \rangle$ and $\langle 4,2 \rangle$ in $T_{2}$, we have $T_{3}$:
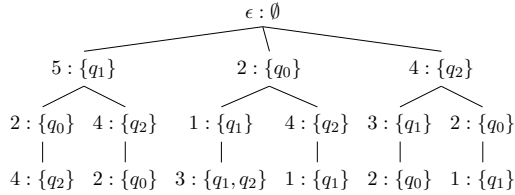
\begin{align*}
\begin{tikzpicture}[scale=0.7]
\Tree [.$\epsilon:\emptyset$
        [.$5:\{q_{1}\}$
          [.$2:\{q_{0}\}$
            [.$4:\{q_{2}\}$ ]
          ]
          [.$4:\{q_{2}\}$
            [.$2:\{q_{0}\}$ ]
          ]
        ]
        [.$2:\{q_{0}\}$
          [.$1:\{q_{1}\}$
            [.$3:\{q_{1},q_{2}\}$ ]
          ]
          [.$4:\{q_{2}\}$
            [.$1:\{q_{1}\}$ ]
          ]
        ]
        [.$4:\{q_{2}\}$
          [.$3:\{q_{1}\}$
            [.$2:\{q_{0}\}$ ]
          ]
          [.$2:\{q_{0}\}$
            [.$1:\{q_{1}\}$ ]
          ]
        ]
      ]
\end{tikzpicture}
\end{align*}
And no more application of $\theta$ is possible. It is not hard to verify that $T_{3}=T(3,4,B')$.
\end{example}

\begin{varlem}{\ref{lem:numeric-bound}~(Numeric Bound).}
$|\calT^{r}(n,k)| = 2^{O(n \lg n)}$.
\end{varlem}
\begin{proof}
Let $T(e,l)$ denote the number of ordered trees with $e$ edges and $l$ leaves. $T(e,l)$ are called \emph{Narayana numbers}, which for $e, l \ge 1$, assume the following closed form~\cite{Nar59,DZ80}:
\begin{align*}
T(e,l)=\frac{1}{e} \binom{e}{l} \binom{e}{l-1} .
\end{align*}
A $\TOP$ has at most $n$ leaves and at most $n \times \mu$ edges. The labels of internal nodes in a $\TOP$ are all determined by the labels of leaves. The number of labels on $l$ leaves is bounded by $l^{n}$, which corresponds to the number of functions from $Q$ to $[1..l]$. Therefore we have
\begin{align*}
|\calT^{r}(n,k)| \le \sum_{e=1}^{n \cdot \mu} \sum_{l=1}^{\min(n,e)}  T(e,l) \cdot l^{n}
               \le n^{3} \cdot T(n^{2}, n) \cdot n^{n}
               \le n \binom{n^{2}}{n} \binom{n^{2}}{n-1} \cdot n^{n}
               = 2^{O(n \lg n)} . \tag*{\qedhere}
\end{align*}
\end{proof} 


\begin{thebibliography}{99}

\bibitem[1]{Buc62} J.R. B\"{u}chi. On a decision method in restricted second
    order arithmetic. In Proc.\textsl{Internat. Congr. Logic, Method. and Philos. Sci.} 
    1960, pages 1-12, Stanford, 1962. Stanford University Press.
    
\bibitem[2]{CZL09} Y. Cai, T. Zhang, and H. Luo. An improved lower bound for the 
    complementation of Rabin automata. In \textsl{Proc. 24th LICS}, pages 167-176, 2009.

\bibitem[3]{CZ11a} Y. Cai and T. Zhang. A Tight lower bound for Streett
    complementation. Manuscript at {\tt arXiv:1102.2963 [cs.LO]}.
    
\bibitem[4]{DZ80} Nachum Dershowitz and Shmuel Zaks. Enumerations of ordered tress. 
    \textsl{Discrete Mathematics}, Vol. 31, No. 1 (1980) 9-28.

\bibitem[5]{FK84} N. Francez and D. Kozen. Generalized fair termination. In \textsl{Proc. 11th POPL}, pages
    46-53, 1984.
    
\bibitem[6]{FKV06} E. Friedgut and O. Kupferman and M.Y. Vardi.
    B\"{u}chi complementation made tighter. 
    \textsl{International Journal of Foundations of Computer Science}, 
    Vol. 17, No. 4 (2006) 851-867.
    
\bibitem[7]{Fra86} N. Francez.
    Fairness. \textsl{Texts and Monographs in Computer Science.} Springer-Verlag, 1986.
    
\bibitem[8]{Kla91} N. Klarlund. Progress measures for complementation of
    omega-automata with applications to temporal logic. In \textsl{Proc. 32th FOCS}, pages
    358-367, 1991.
    
\bibitem[9]{Kup06} O. Kupferman. Avoiding Determinization. In \textsl{Proc. 21th LICS}, pages 243-254, 2006.

\bibitem[10]{Kur94} R.P. Kurshan. Computer aided verification of coordinating processes:
an automata theoretic approach. Princeton University Press, 1994.

\bibitem[11]{KV01} O. Kupferman and M.Y. Vardi. Weak alternating automata are not
    that weak. \textsl{ACM Transactions on Computational Logic}, 2(3): 408-429,
    2001.

\bibitem[12]{KV04} O. Kupferman and M.Y. Vardi. From complementation to certification.
    In \textsl{10th TACAS}, LNCS 2988, pages 591-606, 2004.
    
\bibitem[13]{KV05b} O. Kupferman and M.Y. Vardi. Safraless decision procedures. In \textsl{Proc. 46th FOCS}, pages
    531-540, 2005.

\bibitem[14]{KV05a} O. Kupferman and M.Y. Vardi. Complementation constructions
    for nondeterministic automata on infinite words. In \textsl{Proc. 11th TACAS}, pages
    206-221, 2005.

\bibitem[15]{Lod99} C. L\"{o}ding. Optimal bounds for transformations of
    omega-automata. In \textsl{Proc. 19th FSTTCS}, volume 1738 of \textsl{LNCS}, pages
    97-109, 1999.
    
\bibitem[16]{MH84} S. Miyano and T. Hayashi. Alternating finite automata on $\omega$-words. \textsl{Theoretical Computer Science},
    32(3):321-330, 1984.
    
\bibitem[17]{Mic88} M. Michel. Complementation is more difficult with automata on
    infinite words. \textsl{CNET}, Paris, 1988.
    
\bibitem[18]{Nar59} T. V. Narayana. A Partial Order and Its Applications to Probability Theory.
    \textsl{Sankhy\={a}: The Indian Journal of Statistics (1933-1960)}, Vol. 21, (1959) 91-98.

\bibitem[19]{Pit06} N. Piterman. From Nondeterministic B\"{u}chi and Streett Automata to Deterministic Parity
    Automata. In \textsl{Proc. 21th LICS}, pages 255-264, 2006.
    
\bibitem[20]{RS59} M. O. Rabin and D. Scott. Finite automata and their decision
    problems. \textsl{IBM Journal of Research and Development}, 3:115-125, 1959.

\bibitem[21]{Saf88} S. Safra. On the complexity of $\omega$-automata. In
    \textsl{Proc. 29th FOCS}, pages 319-327, 1988.
    
\bibitem[22]{Saf92} S. Safra. Exponential Determinization for $\omega$-Automata with Strong-Fairness Acceptance Condition. In
    \textsl{Proc. 24th STOC}, pages 275-327, 1992.
    
\bibitem[23]{Sch09} S. Schewe. B\"{u}chi complementation made tight.
    In \textsl{Proc. 26th STACS}, pages 661-672, 2009.
    
\bibitem[24]{SVW87} A. P. Sistla, M.Y. Vardi, and P.Wolper. The complementation 
    problem for B\"{u}chi automata with applications to temporal logic. \textsl{Theoretical Computer Science},
    49:217-327, 1987.

\bibitem[25]{SV89} S. Safra and M.Y. Vardi. On $\omega$-Automata and Temporal Logics. In
    \textsl{Proc. 29th STOC}, pages 127-137, 1989.

\bibitem[26]{Var07} M.Y. Vardi. The B\"{u}chi complementation saga. In
    \textsl{Proc. 24th STACS}, pages 12-22, 2007.

\bibitem[27]{VW86} M.Y. Vardi. and P. Wolper. An automata-theoretic approach to automatic program verification.
    In \textsl{Proc. 1st LICS}, pages 332-334, 1986.

\bibitem[28]{Yan06} Q. Yan. Lower bound for complementation of $\omega$-automata
    via the full automata technique. In \textsl{Proc. 33th ICALP}, volume 4052 of \textsl{LNCS}, pages
    589-600, 2006.

\end{thebibliography}
\end{document}